\documentclass[12pt]{article}
\newif\ifconference 
\conferencefalse
\usepackage{amsmath,amssymb,amsthm,mathtools,multicol,graphicx,float}
\usepackage{lmodern}
\usepackage{tikz}
\usepackage{fbb}
\usepackage[libertine]{newtxmath}
\ifconference 
    \usepackage[margin=1 in,letterpaper]{geometry}
\else
    \usepackage[margin=2cm,a4paper]{geometry}
\fi
\usepackage{keytheorems}
\usepackage{zref-clever}
\zcsetup{
  tlastsep = {, and~},
    lastsep = {, and~}
}
\usepackage{booktabs}
\zcsetup{cap=true}
\usepackage{textpos,euflag}
\usepackage{hyperref}
\usepackage{authblk,enumitem}
\usepackage{lineno}
\usepackage{tcolorbox}
\usetikzlibrary{shapes.geometric, arrows.meta, positioning}
\ifconference
\hypersetup{
    colorlinks=true,
    citecolor=blue,
    linkcolor=blue,
    filecolor=blue,
    urlcolor=blue,
}
\else 
\hypersetup{
    colorlinks=true,
    citecolor=blue,
    linkcolor=blue,
    filecolor=blue,
    urlcolor=blue,
    pdftitle={Branch-width of represented matroids in matrix multiplication time},
    pdfauthor={Mujin Choi, Tuukka Korhonen, and Sang-il Oum}
}
\fi
\newkeytheorem{theorem}[name=Theorem,parent=section,refname={Theorem,Theorems}, Refname={Theorem,Theorems}]
\newkeytheorem{lemma,proposition,corollary,conjecture,definition,remark}[sibling=theorem]
\newcommand\abs[1]{\lvert #1\rvert}
\newcommand{\bw}{\mathsf{bw}}
\newcommand{\pw}{\mathsf{pw}}
\newcommand{\rw}{\mathsf{rw}}
\newcommand{\lrw}{\mathsf{lrw}}
\newcommand{\cutrk}{\mathsf{cutrk}}
\newcommand{\rk}{\mathsf{rk}}

\title{Branch-width of represented matroids in matrix multiplication time}
\ifconference
\date{}\author{}
\else
\author[2,1]{Mujin Choi\thanks{Supported by the Institute for Basic Science (IBS-R029-C1)}}
\affil[1]{Discrete Mathematics Group, Institute for Basic Science (IBS), Daejeon,~South~Korea}
\affil[2]{Department of Mathematical Sciences, KAIST, Daejeon, South~Korea}
\author[3]{Tuukka Korhonen\thanks{Supported by the European Union under Marie Sk\l{}odowska-Curie Actions (MSCA), project no. 101206430, and by the VILLUM Foundation, Grant Number 54451, Basic Algorithms Research Copenhagen (BARC).
}}
\affil[3]{University of Copenhagen, Copenhagen, Denmark}
\author[1,2]{Sang-il~Oum\textsuperscript{\textasteriskcentered}}
\affil[ ]{\small \textit{Email addresses:} 
\texttt{mujinchoi@kaist.ac.kr},
\texttt{tuko@di.ku.dk},
\texttt{sangil@ibs.re.kr}}
\date{\today}
\fi
\begin{document}
\ifconference
\linenumbers
\fi
\maketitle
\ifconference\else
 \begin{textblock}{20}(-0.5, 8.8)
\large\euflag~\parbox[b]{5cm}{\tiny\textcolor{PantoneReflexBlue}{\textsf{\textbf{Co-funded by}}}
\\
\textcolor{PantoneReflexBlue}{{\textbf{\textsf{the European Union}}}}}
 \end{textblock}
\fi
\begin{abstract}
For an $n$-element matroid $M$ given by an $n \times n$ matrix representation over a finite field~$\mathbb F$ and an integer $k$, we present an algorithm with running time $O_{k,\mathbb F}(n^2)+O(n^\omega)$ that either finds a branch-decomposition of~$M$ of width at most $k$, or confirms that the branch-width of~$M$ is more than~$k$, where $\omega < 2.3714$ is the matrix multiplication exponent, and the $O_{k,\mathbb F}(\cdot)$-notation hides factors that depend on $k$ and $\mathbb F$ in a computable manner.
All previous algorithms, including Hlin\v{e}n\'y and Oum [SIAM J. Comput. (2008)] and Jeong, Kim, and Oum [SIAM J. Discrete Math. (2021)], 
have cubic-time bottlenecks.
Moreover, if the input matrix representation is given in standard form, our algorithm runs in $O_{k,\mathbb F}(n^2)$ time, since $O(n^\omega)$ time is only needed for finding a standard form of the input matrix.
When $M$ is given by an $m \times n$ matrix, the overhead for finding a standard form is $O(mn \min(m,n)^{\omega-2})$.

As corollaries, we obtain faster algorithms for rank-width of directed graphs and path-width of matroids represented over a fixed finite field.
Furthermore, we also present an approximation algorithm for finding branch-width that works on infinite fields provided that the input matrix is in standard form and contains a bounded number of distinct values of entries.

To suggest that our algorithm is optimal, we observe that for every field $\mathbb F$, deciding whether the branch-width of a matroid represented over $\mathbb F$ is $0$ is as hard as deciding whether a square matrix over $\mathbb F$ is singular.
Under the assumption that singularity testing requires $\Omega(n^\omega)$-time, this implies that the overhead of $O(n^{\omega})$ is unavoidable.
We also show strengthenings of this observation to rule out some approximations under this assumption.
\end{abstract}

\section{Introduction}\label{sec:intro}

A matroid is a pair $M=(E,\mathcal I)$ of a finite set~$E$ and a collection $\mathcal I$ of subsets of $E$ called \emph{independent} sets such that $\emptyset\in\mathcal I$, all subsets of independent sets are independent, and all maximal independent subsets of a set $X$ have the same size, denoted by $r(X)$, the \emph{rank} function of $M$.
For a matrix $A$ over a field $\mathbb F$ whose columns are indexed by a set $E$, if $\mathcal I$ is the collection of all subsets of $E$ whose corresponding column vectors of $A$ are linearly independent, 
then $M(A)=(E,\mathcal I)$ is a matroid.
If a matroid~$M$ admits a matrix $A$ over a field $\mathbb F$ such that $M=M(A)$, then it is \emph{representable} over~$\mathbb F$.
We say a matroid is \emph{represented} over~$\mathbb F$ if it is given with its matrix representation over~$\mathbb F$.

Branch-width was introduced by Robertson and Seymour~\cite{RS1991} for graphs, hypergraphs, matroids, and more generally, for connectivity functions. 
Here we will describe its definition for matroids.
Let $M$ be a matroid on a finite set $E$ with the rank function $r$.
The \emph{connectivity function} of a matroid $M$ is defined as 
$\lambda_M(X):=r(X)+r(E\setminus X)-r(E)$.
It is well known that $\lambda_M$ is \emph{symmetric}, meaning that $\lambda_M(X)=\lambda_M(E\setminus X)$ for all $X\subseteq E$, and \emph{submodular}, meaning that 
\[ \lambda_M(X)+\lambda_M(Y)\geq \lambda_M(X\cap Y)+\lambda_M(X\cup Y)\]
for all $X,Y\subseteq E$.
Roughly speaking, 
the \emph{branch-width} of a matroid $M$ 
is the minimum  $k$ such that 
its ground set can be recursively partitioned into two parts until 
singletons remain, such that every set arising in the process has $\lambda_M$-value at most $k$.
Such a decomposition, which we call a \emph{branch-decomposition} of width at most $k$, is naturally expressible by a tree of maximum degree~$3$ and a labelling of its leaves by the elements of the ground set of $M$.
A precise definition will be given in \zcref{sec:preliminary}.
Computing the branch-width of a matroid represented over a fixed field
is known to be NP-hard~\cite[Corollary 2]{HM2005}.

Branch-width of graphs was introduced as a tool to study tree-width of graphs.
Many algorithmic problems on graphs, when restricted to graphs of bounded tree-width, admit efficient algorithms. 
Here are two ingredients for such a framework for graphs of bounded tree-width.
\begin{enumerate}[label=\rm(\alph*)]
    \item\label{itm:bodlaender} Efficient algorithms to find a tree-decomposition of small width when the input graph has  tree-width~$k$.
    \item\label{itm:courcelle} Dynamic programming algorithms to solve given problems based on tree-decompositions of graphs.
\end{enumerate}
For \ref{itm:bodlaender}, 
Bodlaender and Kloks~\cite{BK1996} presented a linear-time algorithm to find a tree-decomposition of width at most $k$ for graphs of tree-width at most $k$. 
For \ref{itm:courcelle}, there are many interesting algorithms known for graphs of bounded tree-width based on a tree-decomposition. One important theorem here is the classic theorem of Courcelle~\cite{Courcelle1990}, stating that any graph property that can be described by a monadic second-order formula can be decided in linear time for graphs of bounded tree-width if the input graph is given with its tree-decomposition of width at most~$k$.

To extend this framework to matroids, the definition of branch-width is more convenient, because it is difficult to imagine an analog of vertices in matroids. We remark that Hlin\v{e}n\'y and Whittle~\cite{HW2003} have successfully defined the tree-width of matroids but their definition is more complicated than the definition of branch-width.
We now have analogs for both \ref{itm:bodlaender} and \ref{itm:courcelle}.
\begin{enumerate}[label=\rm(\Alph*)]
    \item \label{itm:hlineny} Hlin\v{e}n\'y~\cite{Hlineny2002} presented an $O_{k,\mathbb F}(n^3)$-time  algorithm to find a branch-decomposition of width at most $O(k)$ for $n$-element matroids of branch-width at most $k$, when the matroid is represented over a finite field $\mathbb F$. 
    Later, Hlin\v{e}n\'y and Oum~\cite{HO2006} and Jeong, Kim, and Oum~\cite{JKO2019} showed $O_{k,\mathbb F}(n^3)$-time algorithms to find a branch-decomposition of width at most~$k$ for $n$-element matroids of branch-width at most~$k$.
    \item \label{itm:hlinenymso} Hlin\v{e}n\'y~\cite{Hlineny2004} showed that if a property of matroids can be described by a monadic second-order formula, then it can be decided in quadratic time for matroids of branch-width at most $k$, if the input matroid is represented over a fixed finite field and is given with its branch-decomposition of width at most~$k$.\footnote{
        Hlin\v{e}n\'y \cite{Hlineny2004} showed that if  the input matroid is given together with a `parse tree' of width at most $k$, then it takes  linear time to decide whether a fixed monadic second-order formula is true. However, to prepare a parse tree from a branch-decomposition of width at most $k$, one would need the quadratic time as in \zcref{thm:transcript}.
    }
\end{enumerate}

We remark that when a matroid is not representable over a finite field, the part \ref{itm:hlinenymso} above fails~\cite[Section~8.3]{Hlineny2004}.
For the part \ref{itm:hlineny}, if we have an oracle for the rank function of the input matroid that runs in time $\gamma$, we can use the algorithm of Korhonen and Oum~\cite{KO2026} to either find, in time $O_k(\gamma n^{2.5} \log^2 n)$, a branch-decomposition of width at most $k$ or conclude that the branch-width is more than $k$.

\subsection{Branch-width with standard representations as an input}

A \emph{standard representation} of a matroid is a matrix representation such that some columns induce an identity matrix as a submatrix.
Note that this implies that the number of rows is at most the number of columns, so the input size is $O(n^2)$.
In contrast, when we speak of a \emph{general representation}, we mean an arbitrary matrix representation, not necessarily one in standard form.
See \zcref{M} for an example.

Our first contribution is to improve the running time of \ref{itm:hlineny} to $O_{k,\mathbb F}(n^2)$ when the input $n$-element matroid is given by a standard representation. 

\getkeytheorem{thm:standard_represented_matroid_branchwidth_approximation}

Jeong, Kim, and Oum~\cite{JKO2019} provided an algorithm based on dynamic programming to find a branch-decomposition of width at most $k$ if it exists, assuming that the input matroid, 
represented over a finite field, 
is given with a branch-decomposition of bounded width. 
Thus, to obtain \zcref{thm:standard_represented_matroid_branchwidth_approximation}, 
we prove the following theorem that provides 
an approximation algorithm to either find a branch-decomposition of width at most $f(k)$
for some function~$f$
or confirm that the branch-width is more than $k$.
We then combine this with \cite{JKO2019}.
Moreover, the approximation algorithm applies 
even when $\mathbb F$ is infinite, as long as the matrix has at most $q-1$ distinct non-zero values of entries for some $q$. 
Note that if $\mathbb F$ is finite, then one can always take $q \le \abs{\mathbb{F}}$.

\begin{figure}
\centering
\small
\(
M\left(
\begin{array}{cccccc}
1 & 2 & 0 & 1 & 1 & 2 \\
0 & 0 & 2 & 2 & 2 & 2 \\
0 & 2 & 0 & 0 & 1 & 1 \\
1 & 0 & 1 & 1 & 1 & 2 
\end{array}\right)
=
M\left(
\begin{array}{ccc|ccc}
1 & 0 & 0 & 1 & 0 & 1 \\
0 & 1 & 0 & 0 & 1 & 1 \\
0 & 0 & 1 & 1 & 1 & 1
\end{array}\right).
\)
\caption{An example of a general representation and a standard representation representing the same matroid.}\label{M}
\end{figure}

\getkeytheorem{thm:standard_represented_matroid_approx}

A \emph{monadic second-order (MSO) logic formula for a matroid} is an MSO formula built up from variables for matroid elements and the predicate indicating whether a set of elements is independent.
We call a matroid property an \emph{MSO property} if there exists an MSO formula describing the property.
We refer to Hlin\v{e}n\'y~\cite{Hlineny2004} for a more detailed explanation of MSO formulas on matroids.
As a corollary of \zcref{thm:standard_represented_matroid_branchwidth_approximation}, we can use the result of Hlin\v{e}n\'y~\cite{Hlineny2004} to deduce that any monadic second-order property of matroids can be decided in quadratic time for matroids of bounded branch-width if the input matroid is representable over a fixed finite field $\mathbb F$ and is given by its standard representation over~$\mathbb F$.
This improves the previously known total running time of $O(n^3)$~\cite{Hlineny2004}, where 
the bottleneck was finding the branch-decomposition.

\begin{corollary}[label=cor:MSO_matroid_standard]
    Let $\mathbb F$ be a finite field, $k$ be an integer, and let $\psi$ be an MSO formula.
    We can decide in time $O_{k,\mathbb F,\psi}(n^2)$ whether $\psi$ is satisfied by an $n$-element matroid that is given by a standard representation over $\mathbb F$ and has branch-width at most $k$.
\end{corollary}

\subsection{Branch-width with general representations as an input}
If the input is given by a general $m\times n$ matrix representation, then
we can convert it into a standard representation in time $O(mn\min(m,n)^{\omega-2})$, where $\omega < 2.3714$ is the matrix multiplication exponent~\cite{ADWXXZ2025}, by using an implementation of Gaussian elimination via fast matrix multiplication~\cite{IMH1982}.
If a matrix has the property that there are at most $q-1$ distinct non-zero values of subdeterminants, then its standard form will have at most $q$ distinct values of entries, see \zcref{lem:subdet}, and therefore we deduce the following from \zcref{thm:standard_represented_matroid_approx}.

\getkeytheorem{thm:branch-decomposition_subdeterminant}

If the underlying field is finite, then we can use \zcref{thm:standard_represented_matroid_branchwidth_approximation} to obtain the following.

\getkeytheorem{thm:represented_matroid_branchwidth_approximation}

As in \zcref{cor:MSO_matroid_standard}, \zcref{thm:branch-decomposition_subdeterminant} implies 
that MSO properties of matroids of branch-width at most $k$
represented over a finite field~$\mathbb F$
can be checked in time $O(mn \min(m,n)^{\omega-2}) + O_{k,\mathbb F,\psi}(n^2)$.

\subsection{Main idea}
Our key idea is to convert our problem into a problem on the rank-width of graphs.
Rank-width of a graph was introduced by Oum and Seymour~\cite{OS2004}.
Given a standard representation~$(~I\mid A~)$ 
of a matroid~$M$
over a finite field $\mathbb F$
with  $q=\abs{\mathbb F}$, 
we construct a graph~$G_A$, which we will call the \emph{entry graph}, 
such that 
if $k$ is the rank-width of~$G_A$
and $b$ is the branch-width of~$M$,
then 
\[ 
\log_q\left(1+\frac{k}{q-1}\right) \le b \le 2^k+k-1.\] 
Our bounds also apply when $\mathbb F$ is infinite; in this case, $q-1$ will be chosen as the number of distinct non-zero values of entries in~$A$.
The following algorithm of Fomin and Korhonen~\cite{FK2024} finds, in quadratic time, a rank-decomposition of width at most~$k$ of a given graph, if one exists.
We apply it to the entry graph~$G_A$.

\begin{theorem}[label=thm:graph_rankwidth_approximation,note={Fomin and Korhonen~\cite{FK2024}}]
    Let $n\ge 2$.
    For an integer $k$, 
    in time $O_k(n^2)$, we can either find a rank-decomposition of the input $n$-vertex graph~$G$ of width at most $k$, or correctly conclude that the rank-width of $G$ is more than $k$.
\end{theorem}

We could also use the almost-linear-time algorithm of Korhonen and Soko\l{}owski~\cite{KS2024Abstract} for computing rank-width, but it would not change the running time we end up with.

Interestingly, previous algorithms due to Hlin\v{e}n\'y and Oum~\cite{HO2006} and Jeong, Kim, and Oum~\cite{JKO2019} for finding a rank-decomposition of small width convert the graph rank-width problem into a problem on matroid branch-width.
However, we convert the matroid branch-width problem back to rank-width problems with explicit bounds.
See \zcref{fig:matroid_algorithm} for a flowchart.

\begin{figure}\centering
\resizebox{1\textwidth}{!}{\begin{tikzpicture}[
    node distance=0.5cm and 0.5cm,
    base/.style={draw=black, thick, align=center, inner sep=0.75ex},
    startstop/.style={base, rectangle, rounded corners, minimum width=3cm, fill=gray!10},
    process/.style={base, align=center, rectangle, minimum width=4cm, text width=7cm, inner sep=0.5ex,
    execute at begin node={\strut}, 
    execute at end node={\strut}},
    decision/.style={base, diamond, aspect=2.5, text width=4.5cm, inner sep=-1.25ex},
    arrow/.style={thick, ->, >=Latex}
    ]

    \node (n1) [startstop] {Input: Matroid $M$ with a matrix representation};
    \node (n2) [process, below=of n1, minimum height=0.8cm] {Convert the matrix to standard form};
    \node (n3) [process, below=of n2, minimum height=1.3cm] {Construct a symmetric matrix $A$\\
    with $\bw(M)=\rw(A)$};
    \node (n4) [process, below=of n3, minimum height=0.8cm] {Construct the entry graph $G_A$};
    \node (n5) [decision, below=of n4] {Approximate\\ rank-width of $G_A$ \\ (\zcref{thm:graph_rankwidth_approximation})};

    \node (n6) [process, below=1.2cm of n5, minimum height=1.3cm] {Convert a rank-decomposition of $G_A$\\ to a rank-decomposition of $A$};
    \node (n7) [decision, below=of n6] {Is rank-width of $A$ at most $k$?\\ (\zcref{thm:branch-width_compression})};
    \node (n8) [process, below=1.2cm of n7, minimum height=1.3cm] {Convert to a branch-decomposition of $M$};
    \node (final) [startstop, below=of n8] {Output: Branch-decomposition of $M$ of width $\leq k$};
    \node (out2) [startstop, right=1cm of final, text width=4.5cm] {Output: $\bw(M)>k$};
    
    \draw [arrow] (n1) -- (n2);
    \draw [arrow] (n2) -- (n3);
    \draw [arrow] (n3) -- (n4);
    \draw [arrow] (n4) -- (n5);
    \draw [arrow] (n5) -| node[above, xshift=-2.8cm] {$G_A$ has large rank-width} ([xshift=0.7cm]out2.north);
    \draw [arrow] (n5) -- node[left, yshift=0.2cm, text width=5cm, align=right] {rank-decomposition \\ of $G_A$ of small width} (n6);
    \draw [arrow] (n6) -- (n7);
    \draw [arrow] (n7) -| node[above, xshift=-2cm, yshift=0.1cm] {rank-width of $A$ is more than $k$} ([xshift=-0.3cm]out2.north);
    \draw [arrow] (n7) -- node[left, yshift=0.2cm, text width=4.5cm, align=right] {rank-decomposition \\ of $A$ of width at most $k$} (n8);
    \draw [arrow] (n8) -- (final);
\end{tikzpicture}}
\caption{Flowchart of the algorithm that, given an input matroid $M$ represented by a matrix over a finite field $\mathbb{F}$, either finds a branch-decomposition of width at most $k$ or concludes that the branch-width of $M$ is greater than $k$.}
\label{fig:matroid_algorithm}
\end{figure}
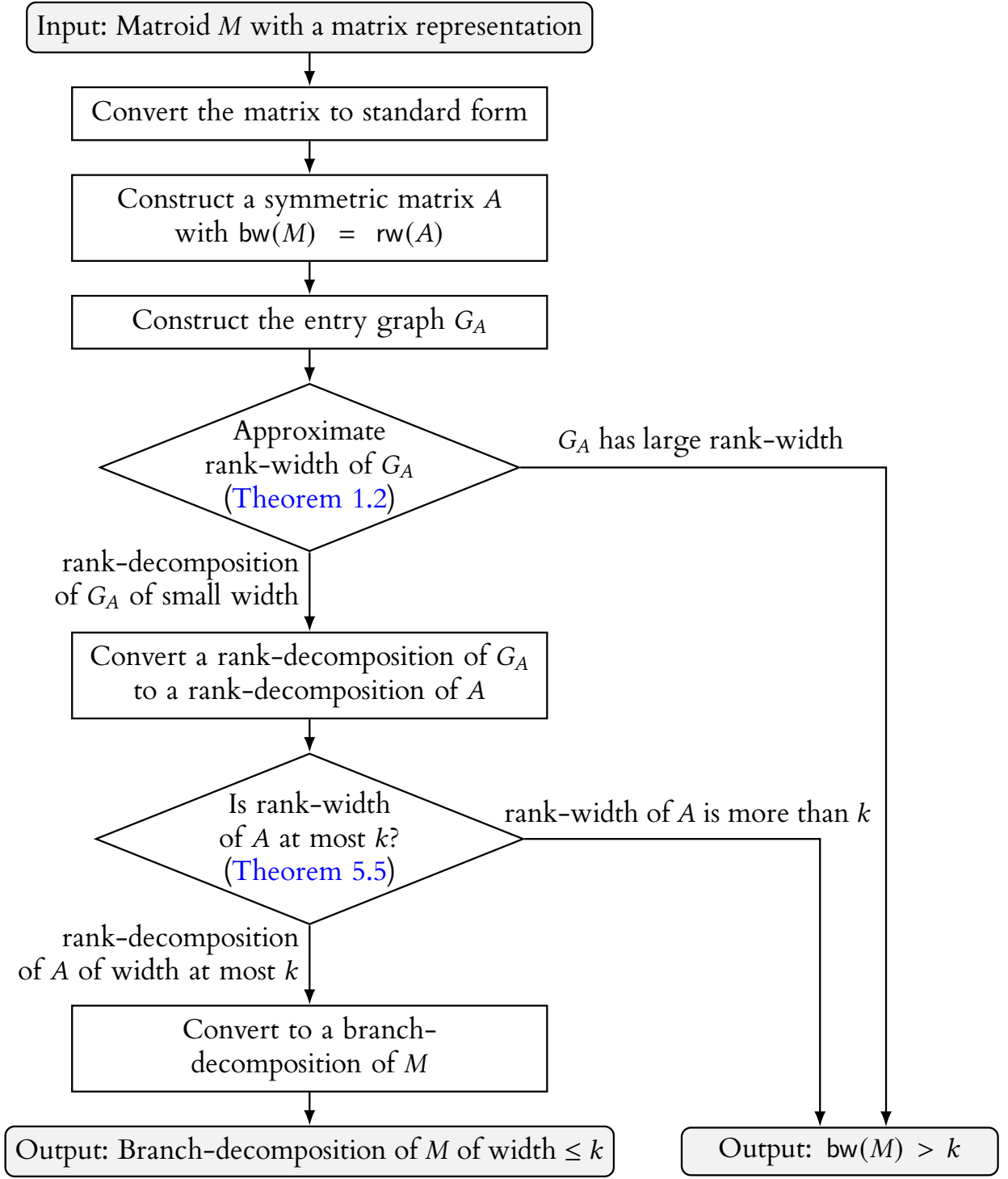

To make our theorem more general, we will consider a more general problem of finding a rank-decomposition of $\sigma$-symmetric matrices, a common generalization of symmetric matrices and skew-symmetric matrices, introduced by Kant\'e and Rao~\cite{KR2013}. 
For a matroid given with a standard matrix representation $(~I \mid A~)$, 
we will construct a symmetric matrix 
$(\begin{smallmatrix}0&A\\A^T&0\end{smallmatrix})$, where $A^T$ denotes the transpose of $A$

\getkeytheorem{thm:exact_algorithm}

The rank-width of a directed graph was introduced by Kant\'e and Rao~\cite{KR2013}.
The above theorem allows us to obtain the following corollary for directed graphs, improving the previous running time of $O(n^3)$ due to Kant\'e and Rao~\cite[Theorem 3.39]{KR2013}.

\getkeytheorem{cor:directedrankwidth}

\subsection{Applications to linear rank-width and path-width}
The \emph{path-width} of a matroid $M$ is the minimum $k$ such that there is an ordering $e_1,e_2,\ldots,e_n$ of the elements of $M$ so that $\lambda_M(\{e_1,e_2,\ldots,e_i\})\le k$ for all $i$.
This is also called the \emph{minimum trellis state-complexity} or \emph{trellis-width} in \cite{Kashyap2008}, motivated by their origin in trellis complexity from coding theory.
It is known that computing the path-width of a matroid represented over a fixed finite field is NP-hard, as shown by Kashyap~\cite{Kashyap2008}.
Jeong, Kim, and Oum~\cite{JKO2016} presented an algorithm to decide whether the path-width of an $n$-element matroid represented over a finite field $\mathbb F$ is at most $k$ in time $O_{k,\mathbb F}(n^3)$.
By using \zcref{thm:represented_matroid_branchwidth_approximation}, we can improve $O_{k,\mathbb F}(n^3)$ to $O_{k,\mathbb F}(n^2)$ if the input matroid is given by a standard representation.
\getkeytheorem{thm:exact_path-width_b}

If the input matroid is given by a general $m\times n$ matrix representation instead, then as before, we can convert it to a standard representation and apply \zcref{thm:exact_path-width_b}.

\getkeytheorem{thm:path-width}

Note that the same method can be applied to linear rank-width of graphs; by combining the result of Fomin and Korhonen (\zcref{thm:graph_rankwidth_approximation}) 
and the result of Jeong, Kim, and Oum~\cite{JKO2016}, 
one can obtain an $O_{k}(n^2)$-time algorithm to decide whether the linear rank-width of an input graph is at most $k$ for each fixed $k$.

\subsection{Lower bound}
Lastly, we give evidence that the factor of $n^{\omega}$ in \zcref{thm:represented_matroid_branchwidth_approximation} is unavoidable when the matroid is given by a general representation.
In particular, we make the simple observation that deciding whether the branch-width of a matroid represented by an $n\times n$ matrix over a field $\mathbb F$ is $0$ is equivalent to deciding whether the  matrix is non-singular.

\getkeytheorem{lem:branchwidth0}

Currently there is an algorithm that decides in time $O(n^\omega)$ whether an $n\times n$ matrix is singular, where $\omega$ is the matrix multiplication exponent~\cite{BH1974,IMH1982}, but the authors are unaware of any algorithm that runs in time $O(n^{\omega'})$ for $\omega'$ less than the current matrix multiplication exponent.
Therefore, \zcref{lem:branchwidth0} implies that an improvement to the algorithm of \zcref{thm:represented_matroid_branchwidth_approximation} for the $n \times n$ case would imply an improvement to a fundamental problem in computational algebra.

We are able to lift the simple observation of \zcref{lem:branchwidth0} to also rule out approximation algorithms.
We show that for a fixed field $\mathbb{F}$ with at least $4$ elements, computing the branch-width of a matroid represented by a matrix over~$\mathbb{F}$ within a factor $1.5$ is as hard as the problem of deciding whether a square matrix over $\mathbb{F}$ is singular.
Note that in the following statement, each of $k$, $\alpha$, and $\mathbb F$ is a fixed constant.

\getkeytheorem{thm:small_field_case}

Therefore, it is unlikely that we can approximate the branch-width of a given 
$n$-element represented matroid over a fixed field $\mathbb{F}$ within a factor $1.5$ in time $O_{\mathbb F}(n^{\omega-\varepsilon})$ for any $\varepsilon>0$ where $\omega < 2.3714$ is the matrix multiplication exponent~\cite{ADWXXZ2025}.
Furthermore, if the size of the field is large enough, we can strengthen the lower bound result to an approximation ratio of~$2-\delta$ for every $\delta>0$.

\getkeytheorem{thm:new_reduction_general}

\subsection{Overview}
This paper is organized as follows.
\zcref{sec:preliminary} will introduce several definitions and lemmas we will use throughout the paper.
In \zcref{sec:matroidtomatrix}, we will find a symmetric matrix whose rank-width is equal to the branch-width of a given represented matroid.
\zcref{sec:approximation} will discuss the key method for the paper, which constructs a graph whose rank-width is functionally equivalent to the rank-width of a $\sigma$-symmetric matrix.
In \zcref{sec:exact}, we will use the algorithm of Jeong, Kim, and Oum~\cite{JKO2019} to convert our approximation algorithm to an exact fixed-parameter tractable algorithm.
In \zcref{sec:linear_rank-width}, we will apply our result to present an algorithm finding a path-decomposition of a represented matroid if one exists.
Finally, in \zcref{sec:reduction}, we will prove the lower bound of running time on approximating the branch-width of matroids.

\section{Preliminaries}\label{sec:preliminary}
All graphs and matroids are finite in this paper. 
For a non-negative integer $n$, let $[n]$ denote the set of positive integers less than or equal to $n$.

An $X \times Y$ matrix, for sets $X$ and $Y$, means a matrix whose rows are indexed by $X$ and columns by $Y$.
For an $X\times Y$ matrix $A$ and subsets $X'\subseteq X$ and~$Y'\subseteq Y$, 
we write $A[X',Y']$ to denote the $X'\times Y'$ submatrix of~$A$.

For a tuple of parameters $\bar{x}$, the $O_{\bar{x}}(\cdot)$-notation hides factors that depend only on $\bar{x}$ and are computable given $\bar{x}$.

\subsection{Rank of matrices}
We prove several lemmas on linear algebra we will use in this paper.
For a matrix $A$, let $\rk(A)$ be the rank of $A$.

\begin{lemma}[label=lem:rank_and_distinct_rows]
    Let $A$ be a matrix that has at most $q-1$ distinct non-zero values of entries.
    Let $d$ be the number of distinct non-zero rows of~$A$.
    Then $\rk(A)\leq d\leq q^{\rk(A)}-1$.
\end{lemma}
\begin{proof}
    The first inequality is trivial.
    To show the second inequality, 
    let $k$ be the rank of $A$.
    The column basis of $A$ has $k$ columns
    and each row vector is determined by its entries on those $k$ columns,
    because every column is a linear combination of the basis columns, so the entry of any column in a given row is the same linear combination of the basis-column entries in that row.
    Each of the $k$ basis-column entries can take at most $q$ values (zero or one of the at most $q-1$ non-zero values), so there are at most $q^k$ distinct row vectors in total, hence at most $q^k-1$ distinct non-zero rows.
\end{proof}

\begin{lemma}[label=lem:matrix_rank_k]
    Let $A$ be a $U\times V$ matrix over the binary field $\mathbb{F}_2$ of rank $k$.
    Then there exists a set $\{(U_i,V_i)\}_{i\in [\ell]}$ of pairs $U_i\subseteq U$ and $V_i\subseteq V$  for some $\ell\in [2^k-1]\cup\{0\}$
    such that $(U_i\times V_i)\cap (U_j\times V_j)=\emptyset$ for all distinct $i,j\in [\ell]$ 
    and the $(a,b)$-entry of $A$ is $1$ if and only if $(a,b)\in U_i\times V_i$ for some $i\in[\ell]$.
\end{lemma}
\begin{proof}
    Let $U_1$, $U_2$, $\ldots$, $U_\ell$ be subsets of $U$ representing the sets of indices of non-zero rows having identical row vectors. By \zcref{lem:rank_and_distinct_rows}, $\ell\le 2^k-1$.
    For each $i\in [\ell]$, let $V_i$ be the subset of $V$ representing the coordinates having $1$ in the row vector of a row indexed by elements of~$U_i$. 
\end{proof}

Let $A$ be an $m\times n$ matrix where $a_{i,j}$ is the $(i,j)$-entry of $A$, and let $B$ be a $p\times q$ matrix.
The \emph{tensor product} of $A$ and $B$, denoted by $A\otimes B$, is the $mp\times nq$ matrix defined by
\[
    A\otimes B=\begin{pmatrix}
        a_{1,1}B & a_{1,2}B & \cdots & a_{1,n}B\\
        a_{2,1}B & a_{2,2}B & \cdots & a_{2,n}B\\
        \vdots & \vdots & \ddots & \vdots \\
        a_{m,1}B & a_{m,2}B & \cdots & a_{m,n}B
    \end{pmatrix}.
\]
It is easy to see that the rank of the tensor product of two matrices is equal to the product of their ranks.
\begin{lemma}[label=lem:tensor_rank]
    Let $A,B$ be matrices.
    Then $\rk(A\otimes B)=\rk(A)\cdot \rk(B)$.
\end{lemma}

A \emph{subdeterminant} of a matrix is the determinant of a square submatrix.

\subsection{Matroids}

For a matroid $M=(E,\mathcal{I})$, a minimal set not in $\mathcal{I}$ is called a \emph{circuit}.
A matroid is called \emph{simple} if each of its circuits contains at least three elements.
Note that if a matroid represented by a matrix $A$ is simple, then $A$ has no zero column and no two columns of $A$ are linearly dependent.

For a matroid $M=(E,\mathcal{I})$ and $E'\subseteq E$, $M\setminus E'$ is the matroid $(E\setminus E',\mathcal{I}')$ where $\mathcal{I}'=\{X\subseteq E\setminus E': X\in \mathcal{I}\}$.
This operation is called the \emph{deletion} of $E'$.
The \emph{restriction} of $M$ to $E'$
is defined as $M\setminus (E\setminus E')$.
It is easy to see that branch-width of a matroid does not increase when deleting an element.

\begin{lemma}[\cite{Oxley2011M}]\label{lem:bw_deletion}
    Let $M=(E,\mathcal{I})$ be a matroid and let $E'\subseteq E$.
    Then $\bw(M\setminus E')\leq\bw(M)$.
\end{lemma}

A matroid~$M$ is \emph{uniform} if every subset of size $\rk(M)$ is independent.
The branch-width of a uniform matroid can be computed as follows.
\begin{proposition}[note={Dharmatilake~\cite[Proposition 2.1.5(5)]{Dharmatilake1994}}]\label{lem:uniformbw}
    The branch-width of a uniform matroid~$M$ is equal to  
    \[ \min(\rk(M),\abs{E(M)}-\rk(M),\lceil\abs{E(M)}/3\rceil).\]
\end{proposition}

For more about matroid theory, we refer to the book of Oxley~\cite{Oxley2011M}.

\subsection{Connectivity functions and branch-width}
A \emph{connectivity function} on a finite set $E$ is a submodular function $\lambda:2^E\to\mathbb Z_{\geq0}$ such that 
$\lambda(X)=\lambda(E\setminus X)$ for all $X\subseteq E$
and $\lambda(\emptyset)=0$.
The matroid connectivity function is an instance of a connectivity function.

A tree is \emph{subcubic} if every node has degree at most $3$.
A \emph{branch-decomposition} of a connectivity function $\lambda$ is a pair $(T,\mathcal{L})$ of a subcubic tree $T$ and a bijection $\mathcal{L}:E\rightarrow \{t:t\text{ is a leaf of }T\}$.
For each edge $e$ of $T$, the connected components of~$T\setminus e$ induce a partition $(X_e,Y_e)$ of $E$ given by the leaves of $T$ in each component, and the \emph{width} of an edge $e$ is defined by $\lambda(X_e)$.
Next, the \emph{width} of a branch-decomposition~$(T,\mathcal{L})$ is defined by the maximum width among the edges of $T$.
Lastly, the \emph{branch-width} of $\lambda$, denoted by $\bw(\lambda)$, is the minimum width of all its branch-decompositions.
If $\abs{E}\leq 1$, then $\lambda$ admits no branch-decomposition and we define $\bw(\lambda)=0$.

Here is a lemma on a subcubic tree, which we will use later for branch-decompositions.

\begin{lemma}\label{lem:balancedcut}
    Let $k\ge2$ be an integer.
    Every subcubic tree $T$ with at least $3k-2$ leaves
    has an edge $e$ such that 
    each component of $T\setminus e$ contains at least $k$ leaves of $T$.
\end{lemma}
\begin{proof}
    Suppose not. 
    Then we can orient each edge $e=uv$ of $T$ toward~$v$ 
    if the component of~$T\setminus e$ containing $u$ has fewer than $k$ leaves.
    By following along oriented edges, we are guaranteed to find a node~$v$ of~$T$ such that every edge incident with~$v$ is oriented toward $v$.
    If $v$ is a leaf of~$T$, then $T$ has at most~$(k-1)+1\le 3k-3$ leaves, contradicting the assumption.
    If $v$ is not a leaf of~$T$, then $v$ has degree at most $3$ and therefore $T$ has at most $3(k-1)$ leaves, a contradiction.
\end{proof}

Later, we will sometimes discuss a rooted branch-decomposition. 
For this, we relax our definition of a subcubic tree to allow the root to have degree~$2$.
A \emph{rooted} branch-de\-com\-po\-si\-tion is a branch-decomposition $(T,\mathcal L)$ where $T$ is a tree having a designated root~$r$ of degree~$2$.

As we discussed earlier, the branch-width of a matroid is exactly the branch-width of its connectivity function. 
There are other connectivity functions arising from graphs, giving rise to various definitions of width parameters of graphs. Here are examples.
\begin{description}
    \item [Branch-width]
For a subset~$X$ of edges of a graph~$G$, let~$\kappa_G(X)$ be the number of the vertices of~$G$ incident with both an edge in $X$ and an edge in $E(G)\setminus X$. It is well known that~$\kappa_G$ is a connectivity function as well 
and 
the branch-width of~$\kappa_G$ is the \emph{branch-width of a graph~$G$}, introduced by Robertson and Seymour~\cite{RS1991}.
\item [Rank-width]
For a subset $X$ of vertices of a graph~$G$, let 
$\cutrk_G(X)=\rk(A_G[X,V(G)\setminus X])$
where $A_G$ is the adjacency matrix of $G$ over the binary field.
This function is called the \emph{cut-rank} function 
and is a connectivity function \cite{OS2004}.
The branch-decomposition of $\cutrk_G$ is the definition of a \emph{rank-decomposition} of a graph $G$
and the branch-width of $\cutrk_G$ is the definition of \emph{rank-width of a graph $G$}, introduced by Oum and Seymour~\cite{OS2004}.

\end{description}

\subsection{Rank-width of \texorpdfstring{$\sigma$}{sigma}-symmetric matrices}
The concept of rank-width of graphs was extended to symmetric or skew-symmetric matrices by Oum~\cite{Oum2006b}, and further extended to \emph{$\sigma$-symmetric matrices} by Kant\'e and Rao~\cite{KR2013}.

For a field $\mathbb{F}$, we call $\sigma:\mathbb{F}\rightarrow \mathbb{F}$ a \emph{sesqui-morphism} if $\sigma(\sigma(x))=x$ for all $x \in \mathbb{F}$ and the mapping $[x\mapsto \sigma(1)^{-1}\sigma(x)]$ is an automorphism.
A matrix $A$ is called \emph{$\sigma$-symmetric} if there is a sesqui-morphism $\sigma$ such that $A_{ij}=\sigma(A_{ji})$ for every $i$, $j$.
From the definition, it is easy to see that for two matrices $A$ and $A'$ with $A'_{ij}=\sigma(A_{ij})$ for all $i$ and $j$, we have $\rk(A)=\rk(A')$.
Note that if $\sigma(x)=x$ for all $x$, then $\sigma$-symmetric matrices are symmetric matrices and if $\sigma(x)=-x$ for all~$x$, then $\sigma$-symmetric matrices are skew-symmetric matrices.
In our paper, in most cases, we will only consider the case where $\sigma$ is an identity function.

Now we define the rank-width of a $\sigma$-symmetric matrix.
Let $A$ be an $n\times n$ $\sigma$-symmetric matrix over a field $\mathbb{F}$.
Let $\cutrk_A:2^{[n]}\rightarrow \mathbb{Z}_{\geq 0}$ be defined by
$$\cutrk_A(X)=\rk(A[X,[n]\setminus X])$$
for each $X\subseteq [n]$.
The following lemma shows that the cut-rank function on a $\sigma$-symmetric matrix is symmetric and submodular, which allows us to define the branch-width of $\sigma$-symmetric matrices.

\begin{lemma}[Kant\'e and Rao \cite{KR2013}]
    The cut-rank function $\cutrk_A$ of a $\sigma$-symmetric matrix is symmetric and submodular.
\end{lemma}

The \emph{rank-width} of $A$, denoted by $\rw(A)$, is the branch-width of $\cutrk_A$.
A \emph{rank-decomposition} of $A$ is a branch-decomposition of $\cutrk_A$.
Note that for a graph $G$, the rank-width of $G$ is equal to the rank-width of its adjacency matrix over the binary field.

\subsection{Partitioned matroids and subspace arrangements}
A pair $(M,\mathcal{P})$ is called a \emph{partitioned matroid} if $M=(E,\mathcal{I})$ is a matroid and $\mathcal{P}$ is a partition of $E$.
A connectivity function of a partitioned matroid $(M,\mathcal{P})$, denoted by $\lambda_M^\mathcal{P}$, is defined by 
\[\lambda_M^{\mathcal{P}}(X)=\lambda_M\big(\bigcup_{Y\in X}Y\big)\]
for all $X\subseteq\mathcal P$.
Note that $\lambda_M^\mathcal{P}$ is indeed a connectivity function.
The branch-width of $(M,\mathcal{P})$, denoted by $\bw(M,\mathcal{P})$, is the branch-width of $\lambda_M^{\mathcal{P}}$.

Another similar notion for partitioned matroid is \emph{subspace arrangement}.
Let $\mathbb F$ be a field and $r$ be a positive integer.
A \emph{subspace arrangement} in $\mathbb F^r$ is a multiset 
$\mathcal V=\{V_1,V_2,\ldots,V_n\}$ of subspaces of~$\mathbb F^r$. 
For a subset $X$ of $[n]$, we define 
\[ 
    f(X)= \dim \Big( \big(\sum_{i\in X}V_i\big)\cap \big(\sum_{j\in [n]\setminus X}V_j\big) \Big). 
\]
Then $f$ is submodular because $X\mapsto\dim\big(\sum_{i\in X}V_i\big)$ is a submodular function on $[n]$ and $f(X)=\dim\big(\sum_{i\in X}V_i\big)+\dim \big(\sum_{j\in [n]\setminus X}V_j\big) - \dim\sum_{i=1}^n V_i$.
Thus $f$ is a connectivity function. 
The \emph{branch-width} of a subspace arrangement $\mathcal{V}$ is the branch-width of $f$ as a connectivity function.

Note that we can regard a subspace arrangement as a partitioned matroid.
This can be done by choosing a basis $E_i$ for each $V_i$ and making a matrix $A$ whose columns are elements of $E_i$ for each $i$.
Let $E$ be the index set of columns of $A$ and let $\mathcal{P}=\{E_i\colon V_i\in \mathcal{V}\}$ and $M$ be the matroid on $E$ represented by $A$.
Then the connectivity function $f$ of $\mathcal{V}$ is equal to the connectivity function~$\lambda_M^\mathcal{P}$ of the partitioned matroid $(M,\mathcal{P})$.

\subsection{Linear layouts and path-width}
Let $\lambda$ be a connectivity function on the subsets of a finite set~$E$.
Let $n=\abs{E}$.
A \emph{linear layout} of $\lambda$ is an ordering $e_1,e_2,\ldots,e_n$ of elements of $E$.
The width of a linear layout $e_1,e_2,\ldots,e_n$ is given by
\[\max_{1\leq i<n}\lambda\big( \{e_1,e_2,\ldots,e_i\} \big).\]
The \emph{path-width} of $\lambda$, denoted by $\pw(\lambda)$, is the minimum width of all its linear layouts.
If $\abs{E}=1$, then we define the path-width to be $0$.

If a connectivity function $\lambda$ satisfies $\lambda(\{e\})\leq 1$ for all $e\in E$, then the path-width of $\lambda$ is an upper bound of the
branch-width of $\lambda$.

\begin{lemma}[label=lem:lrwleqrw]
    Let $\lambda$ be a connectivity function on the subsets of a finite set $E$.
    If $\lambda(\{e\})\leq 1$ for all $e\in E$, 
    then 
    the branch-width of $\lambda$ 
    is less than or equal to the path-width of $\lambda$.
\end{lemma}
\begin{proof}
    Let $k=\pw(\lambda)$.
    We may assume that $n=\abs{E} \ge 3$ because otherwise it is trivial.
    Let~$e_1,e_2,\ldots,e_n$ be a linear layout of width at most $k$.
    Let $T$ be a tree on $2n-2$ vertices $v_1,v_2,\ldots,v_n,w_2,\ldots,w_{n-1}$
    where $E(T)=\{v_iw_i : 2 \le i \le n-1\} 
    \cup \{ w_i w_{i+1}: 2 \le i \le n-2\} \cup \{v_1 w_2, v_n w_{n-1}\}$.
    Let~$\mathcal L(e_i)=v_i$ for $i\in\{1,2,\ldots,n\}$.
    We claim that $(T,\mathcal L)$ is a branch-decomposition of width at most $k$. 
    For each~$i\in\{2,\ldots,n-2\}$, 
    the width of an edge $w_iw_{i+1}$ is equal to 
    $\lambda(\{e_1,e_2,\ldots,e_i\})$, which is at most $k$ because the linear layout has width at most~$k$.

    Note that all the other edges of $(T, \mathcal L)$ have width of form $\lambda(\{e_i\})$ for $i \in [n]$.
    We claim that $\lambda(\{e_i\})\le k$ for all $i\in\{1,2,\ldots,n\}$.
    Since $\lambda(\{e_i\})\le 1$, we may assume that~$k=0$.
    Now, we have $\lambda(\{e_i,e_{i+1},\ldots,e_n\})=\lambda(\{e_1,e_2,\ldots,e_{i-1}\})=0$ 
    and $\lambda(\{e_1,e_2,\ldots,e_i\})=0$. 
    By the submodularity, 
    $\lambda(\{e_i\})+\lambda(E)\le \lambda(\{e_1,e_2,\ldots,e_i\}) + \lambda(\{e_i,e_{i+1},\ldots,e_n\})= 0$ 
    and therefore $\lambda(\{e_i\})=0$.
\end{proof}

A \emph{path-decomposition} of a matroid is a linear layout of its connectivity function, and the \emph{path-width} of a matroid $M$, denoted by $\pw(M)$, is defined as the path-width of its connectivity function.
Similarly, a \emph{linear rank-decomposition} of a graph is a linear layout of its cut-rank function, and the \emph{linear rank-width} of a graph, denoted by $\lrw(G)$, is defined as the path-width of its cut-rank function.
A \emph{linear rank-decomposition} of a $\sigma$-symmetric matrix is also a linear layout of its cut-rank function, and the \emph{linear rank-width} of a $\sigma$-symmetric matrix $A$, denoted by $\lrw(A)$, is defined as the path-width of its cut-rank function.

\subsection{Matrix multiplication exponent}

Let $\omega$ be a constant such that 
multiplying two $n\times n$ matrices over any field can be performed by 
$O(n^{\omega})$ arithmetic operations. 
Then $\omega\ge 2$. Furthermore, we can choose $\omega < 2.3714$ due to~\cite
{ADWXXZ2025}.

\section{Converting represented matroids to symmetric matrices}\label{sec:matroidtomatrix}
When dealing with matroids represented over a field,
it is more convenient to work with standard representations.
When the input matroid is given with a general representation, we need to convert it into a standard representation. 
Here is a lemma on this conversion. 
We would like to state our main theorem in full generality, so instead of restricting our attention to a finite field, we will use an upper bound on the number of distinct values of entries in a standard matrix representation.
For general representations, this can be translated into an upper bound on the number of distinct values of subdeterminants.

\begin{lemma}[label=lem:subdet]
    Let $M$ be a matroid represented by a matrix $A$ over a field $\mathbb F$.
    Let $K$ be a subset of $\mathbb F$ such that $-K=K$.
    If every non-zero subdeterminant of~$A$ belongs to~$K$, then 
    there exists a non-zero element $\alpha\in\mathbb F$ 
    and a standard representation of $M$ 
    whose non-zero entries belong to $\alpha K$.
    Furthermore, if $A$ has $m$ rows and $n$ columns, then such $\alpha$ and such a standard representation can be found in time $O(mn\min(m,n)^{\omega-2})$.
\end{lemma}
\begin{proof}
    Let $r \le \min(n,m)$ be the rank of~$M$.
    First, we find $r$ linearly independent rows of $A$ in time $O(mn \min(m,n)^{\omega-2})$ by the algorithm of~\cite{IMH1982}.
    We let $A'$ be the $r \times n$ matrix obtained by deleting all other rows of $A$.
    The matrix $A'$ represents the same matroid as the matrix $A$, because deleting linearly dependent rows does not affect whether a set of columns is linearly independent.

    Then, we use the algorithm of~\cite{IMH1982} to find in time $O(r^{\omega-1} n)$ a basis $B$ of $r$ linearly independent columns of $A'$.
    Let $P$ be the submatrix of $A'$ obtained by taking the columns in $B$.
Then $A'' = P^{-1} A'$ is a standard representation of $M$ with respect to $B$.
    The matrix $A''$ can be computed in $O(r^{\omega-1} n)$ time~\cite{BH1974}.

    It remains to show that there exists a non-zero element $\alpha \in \mathbb{F}$ so that all non-zero entries of $A''$ belong to $\alpha K$.
    We show that we can take $\alpha=\det P^{-1}$.

    Let $a_{ij}''$ be the $(i,j)$ entry of $A''$.
    We assume that the rows of $A''$ are indexed by $B$ so that $a_{ii}''=1$ for $i\in B$.
    If $j\in B$, then $a_{ij}''\in\{0,1\}$  because $P^{-1} P= I$.
    We have $1\in \alpha K$ because $\alpha\det P=1$.
    If $j\notin B$, then 
    $a_{ij}''=\pm\det (A''[B,(B\setminus\{i\})\cup\{j\}]) = \pm \det P^{-1} \det A'[B,(B\setminus\{i\})\cup\{j\}]
    \in \alpha K$.
\end{proof}

To make our main proof most general, we will mainly work with $\sigma$-symmetric matrices. 
Let us discuss why problems on branch-width of matroids can be converted 
to problems on rank-width of symmetric matrices. 
As symmetric matrices are instances of $\sigma$-symmetric matrices, 
all the theorems we prove later for $\sigma$-symmetric matrices will give results for matroids by the following lemmas.
We write $I_r$ for the $r\times r$ identity matrix
and $0$ for the zero matrix of certain size.

The following lemma is essentially in the paper of Oum and Seymour~\cite[Proposition 3.1]{OS2004}, where it was proved for binary matroids only.
The proof is identical but for completeness of this paper, we include the proof.

\begin{lemma}[label=lem:matroid_to_symmetric_matrix]
    Let $M=(E,\mathcal{I})$ be a matroid with a standard representation $A=(I_r \mid A_0)$ over a field~$\mathbb{F}$.
    Let $A'$ be an $E\times E$ matrix \[A'= \begin{pmatrix} 0 & A_0 \\ A_0^T & 0 \end{pmatrix}\] over $\mathbb F$.
    Then for all $X\subseteq E$, we have 
    $\lambda_M(X)=\cutrk_{A'}(X)$.
    In particular, $\bw(M)=\rw(A')$.
\end{lemma}
\begin{proof}
Let $F$ be a subset of $E$ 
    corresponding to $I_r$ of $A$, and 
    let $s=\abs{F\cap X}$.
    We may assume that the rows of $A$ are indexed by $F$, by assuming that the $(i,i)$ entry of $A$ is $1$ for all $i\in F$.
    Let $B=A[F\setminus X,X\setminus F]$, $B'=A[F\cap X,X\setminus F]$,
    $C'=A[F\setminus X,E\setminus (X\cup F)]$, and $C=A[F\cap X,E\setminus (X\cup F)]$. 
    Then 
    \[ 
        A[F,X]= 
        \begin{pmatrix}
            0 & B \\
            I_s & B' 
        \end{pmatrix}
        \text{ and }
        A[F,E\setminus X]=
            \begin{pmatrix}
                I_{r-s} & C' \\
                0 & C 
            \end{pmatrix}.
    \]
Therefore, $\lambda_M(X)=\rk(A[F,X])+\rk(A[F,E\setminus X])-r
        =s+\rk(B)+(r-s)+\rk(C)-r
        =\rk(B)+\rk(C)$.
    
    Now observe that 
    \[
        \cutrk_{A'}(X)= \rk (A'[X,E\setminus X])=
            \rk 
            \begin{pmatrix}
                0 & C \\
                B^T & 0 
            \end{pmatrix}
            =\rk(B)+\rk(C)=\lambda_M(X).  \qedhere
    \]
\end{proof}

Note that constructing $A'$ from $A$ can be done in time $O(|E|^2)$.

\section{Rank-width of the entry graph of a \texorpdfstring{$\sigma$}{sigma}-symmetric matrix}\label{sec:approximation}

We now define a graph called the entry graph from a $\sigma$-symmetric matrix $A$ such that the graph has bounded rank-width if and only if $A$ has bounded rank-width. 
This would give an algorithm that allows us to find a rank-decomposition of bounded width for the input $\sigma$-symmetric matrices of rank-width at most $k$
and thus give a corresponding algorithm for matroid branch-width as well. 
This process only works if the $\sigma$-symmetric matrix has a bounded number of distinct non-zero values of entries. This requirement is easily satisfied if the underlying field is finite.

The \emph{entry graph} of a $\sigma$-symmetric matrix $A$ is, roughly, 
the disjoint union of stars together with edges joining leaves of distinct stars
in such a way that we can read the $(u,v)$-entry of $A$ by reading the adjacency between leaves of stars corresponding to $u$ and $v$. 
Since $A$ is $\sigma$-symmetric, this backward translation is possible without causing any ambiguity.
More precisely, the \emph{entry graph} of a $V\times V$ $\sigma$-symmetric matrix $A$, denoted by $G_A$, is a graph on the vertex set 
\[ V\cup \{ v_i^{\alpha} : \text{there exists }j\in V\setminus\{i\}\text{ such that } \alpha=A_{ij}\neq 0\},\]
where $v_i^{\alpha}$ are new distinct vertices not in $V$,
such that
\begin{itemize}
\item every $i\in V$ is adjacent to $v_i^{\alpha}$ 
for all $\alpha\in \{ A_{ij}: j\in V\setminus\{i\},~A_{ij}\neq0\}$
and 
\item $v_i^\alpha$ is adjacent to $v_j^\beta$ if and only if $i\neq j$,  
$\alpha=A_{ij}\neq0 $, and $\beta=A_{ji}\neq0$.
\end{itemize} 
See \zcref{fig:entrygraph} for an illustration.
For each $i\in V$, let \[ W_i=\{i\}\cup \{ v_{i}^{\alpha}: \text{there exists }j\in V\setminus\{i\} \text{ such that }\alpha=A_{ij}\neq0\}.\]
Then $V(G_A)$ is the disjoint union of $W_i$'s for all $i\in V$ and each $W_i$ induces a star 
whose number of leaves is equal to the number of distinct non-zero values in the corresponding row or column.
\begin{figure}
    \centering
\begin{minipage}[c]{0.4\textwidth}

\[A=\ \bordermatrix{
    ~ & a & b & c & d & e \cr
    a & 0 & 0 & 1 & 2 & 0 \cr
    b & 0 & 0 & 1 & 1 & 2 \cr
    c & 1 & 1 & 0 & 0 & 0 \cr
    d & 2 & 1 & 0 & 0 & 2 \cr
    e & 0 & 2 & 0 & 2 & 0 \cr
    }\]
\end{minipage} \
\begin{minipage}[c]{0.4\textwidth}
\begin{tikzpicture}[
    baseline=(current bounding box.center),
    scale=0.4,
    node/.style={circle, fill=black, inner sep=0pt, minimum size=5pt},
    edge/.style={line width=1.5pt, black, line cap=round, line join=round}
]

\coordinate (N1)  at (0.000, 5.000);
\node[node, label={above:$a$}] at (N1) {};
\coordinate (N2)  at (-0.920, 3.250);
\node[node, label={[xshift=0.1cm]left:$v_a^2$}] at (N2) {};
\coordinate (N3)  at (0.920, 3.250);
\node[node, label={[xshift=-0.08cm]right:$v_a^1$}] at (N3) {};
\coordinate (N4)  at (-3.050, 1.700);
\node[node, label={[yshift=-0.1cm]above:$v_b^1$}] at (N4) {};
\coordinate (N5)  at (-4.820, 1.170);
\node[node, label={[xshift=0.05cm,yshift=0.1cm]left:$b$}] at (N5) {};
\coordinate (N6)  at (-3.680, -0.280);
\node[node, label={[yshift=0.1cm]below:$v_b^2$}] at (N6) {};
\coordinate (N7)  at (-2.170, -3.520);
\node[node, label={[xshift=-0.08cm, yshift=-0.05cm]right:$v_c^1$}] at (N7) {};
\coordinate (N8)  at (-2.900, -4.750);
\node[node, label={[xshift=-0.08cm]below:$c$}] at (N8) {};
\coordinate (N9)  at (1.350, -4.050);
\node[node, label={[xshift=0.05cm,yshift=-0.05cm]left:$v_d^1$}] at (N9) {};
\coordinate (N10) at (2.950, -4.750);
\node[node, label={[yshift=0.05cm,xshift=0.05cm]below:$d$}] at (N10) {};
\coordinate (N11) at (2.880, -2.920);
\node[node, label={[xshift=-0.05cm]right:$v_d^2$}] at (N11) {};
\coordinate (N12) at (3.520, 0.740);
\node[node, label={[yshift=-0.05cm]above:$v_e^2$}] at (N12) {};
\coordinate (N13) at (4.820, 1.180);
\node[node, label={[xshift=-0.05cm,yshift=0.1cm]right:$e$}] at (N13) {};

\draw[edge] (N1) -- (N2);
\draw[edge] (N1) -- (N3);
\draw[edge] (N2) -- (N11); 
\draw[edge] (N3) -- (N7);
\draw[edge] (N4) -- (N5);
\draw[edge] (N4) -- (N7);
\draw[edge] (N4) -- (N9);  
\draw[edge] (N5) -- (N6);
\draw[edge] (N6) -- (N12);
\draw[edge] (N7) -- (N8);
\draw[edge] (N9) -- (N10);
\draw[edge] (N10) -- (N11);
\draw[edge] (N11) -- (N12);
\draw[edge] (N12) -- (N13);

\foreach \n in {1,...,13} {
    \node[node] at (N\n) {};
}

\end{tikzpicture}
\end{minipage}
\caption{An example of a symmetric matrix $A$ (left) and its entry graph $G_A$ (right).}\label{fig:entrygraph}
\end{figure}
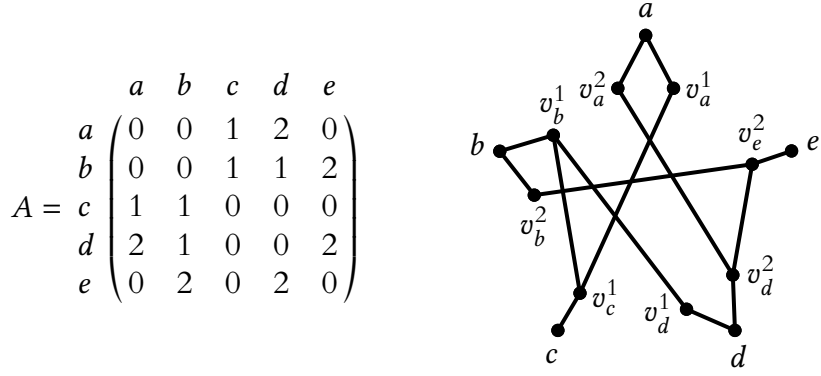

Now we show that
for a $\sigma$-symmetric matrix~$A$, 
if $G_A$ has rank-width $k$, 
then the rank-width of $A$ is at most $2^k+k-1$.
The main idea of the proof is that 
for a vertex set $X$ of the entry graph~$G_A$,
each star crossing $X$ 
contributes $1$ to the cut-rank of $X$ in $G_A$.

\begin{proposition}[label=thm:graph_to_matrix]
    Let $\mathbb F$ be a field and 
    let $\sigma:\mathbb{F}\rightarrow \mathbb{F}$ be a sesqui-morphism.
    Let $V$ be a finite set with at least two elements and let $A$ be a $V\times V$ $\sigma$-symmetric matrix over $\mathbb{F}$.
    If the entry graph $G_A$ of~$A$ has rank-width~$k$, then $A$ has rank-width at most $2^k+k-1$.
Furthermore, a rank-decomposition of the entry graph of~$A$ of width at most $k$ can be converted into a rank-decomposition of~$A$ of width at most $2^k+k-1$ in time $O(\abs{V(G_A)})$.
\end{proposition}
\begin{proof}
    Let $B$ be the adjacency matrix of the entry graph $G_A$ of~$A$, let $k=\rw(G_A)$, and let $(T,\mathcal{L})$ be a rank-decomposition of $G_A$ witnessing $\rw(G_A)$.
    Note that $B$ is a matrix over the binary field~$\mathbb{F}_2$, not over $\mathbb{F}$.

    Let $T'$ be the minimal subtree of $T$ containing all leaves corresponding to $V$ by $\mathcal L$.
    Let~$\mathcal L'$ be the restriction of $\mathcal L$ to $V$.
    Then $(T',\mathcal L')$ is a rank-decomposition of $A$.
We claim that the width of $(T',\mathcal{L}')$ is at most $2^{k}+k-1$.
    Note that $T'$ can be found in time proportional to the number of nodes of $T$, which is $O(\abs{V(G_A)})$. 
    When we output the rank-decomposition of $A$, we can suppress all degree-$2$ nodes of $T'$ to reduce the size of the output.

    Let $e\in E(T')$.
    Let $(X_e,Y_e)$ be the partition of $V(G_A)$ 
    induced by $\mathcal L$ and the connected components of $T\setminus e$.
    Note that the width of $e$ in $(T',\mathcal{L}')$ is equal to $\rk(A[V\cap X_e,V\cap Y_e])$.
    Let us write $N(v)$ to denote the set of all neighbors of $v$ in $G_A$.
    Now let
    \[
        U_1= \{v\in V\cap X_e: N(v)\cap Y_e\neq\emptyset\} \text{ and }
        U_2= \{v\in V\cap Y_e: N(v)\cap X_e\neq\emptyset\}.
    \]

    As the vertices in $U_1$ are in $V$, for each $i \in U_1$ there exists $v_i^{\alpha} \in Y_e$ so that $i$ is adjacent to $v_i^{\alpha}$.
    In particular, there exists a set $U_1' \subseteq Y_e$, so that $\abs{U_1}=\abs{U_1'}$ and $\abs{N(v)\cap U_1'}=1$ for every $v\in U_1$.
    Similarly, there exists $U_2' \subseteq X_e$, so that $\abs{U_2}=\abs{U_2'}$ and $\abs{N(v)\cap U_2'}=1$ for every $v\in U_2$.
    All the four sets $U_1$, $U_2$, $U_1'$, and $U_2'$ are pairwise disjoint.

    Now,
    \[ 
        k\ge \cutrk_{G_A}(X_e) \ge
        \rk \,\,\bordermatrix{
        & U_2 & U_1'\cr
        U_1 & 0 & I \cr
        U_2'& I & *
        } = \abs{U_1}+\abs{U_2},
    \] 
    where $I$ is the identity matrix, $0$ is the zero matrix, and $*$ represents an arbitrary matrix.

    For $S\subseteq V(G_A)$, 
    let us write $N[S]$ to denote the set of all vertices in $S$ and all neighbors of a vertex in $S$ in $G_A$.
    Let
    \begin{align*}
        W_1&= (V \cap X_e)\setminus U_1, &
        W_2&= (V \cap Y_e)\setminus U_2, \\
        W_1'&= N[W_1],&
        W_2'&= N[W_2].
    \end{align*}
    By definition of $U_1$ and $U_2$, we know that $W_1'\subseteq X_e$ and $W_2'\subseteq Y_e$.
    Let $A'=A[W_1,W_2]$ and $B'=B[W_1',W_2']$.
    Since $A'$ can be obtained from $A[V\cap X_e,V\cap Y_e]$ by removing rows indexed by~$U_1$ and columns indexed by~$U_2$, we have
    \[ \rk(A[V\cap X_e,V\cap Y_e])\leq \rk(A')+\abs{U_1}+\abs{U_2}\le \rk(A')+k.\]
    Since $\rk(B')\leq \rk(B[X_e, Y_e])\leq k$, by using \zcref{lem:matrix_rank_k} we can find $Q_1,\ldots,Q_{\ell}\subseteq W_1'$ and $R_1,\ldots,R_{\ell}\subseteq W_2'$ for some $\ell\le 2^{k}-1$ such that $(Q_i\times R_i)\cap (Q_j\times R_j)=\emptyset$ for all distinct $i,j\le \ell $ and the $ab$-entry of~$B'$ is $1$ if and only if $(a,b)\in Q_i\times R_i$ for some $i\le \ell$.
    
    Now we describe how $A'$ can be obtained from~$B'$.
    First note that all rows or columns of~$B'$ corresponding to $W_1$ or $W_2$ are zero.
    For each $i\le\ell$, let $A_i$ be a $W_1\times W_2$ matrix over $\mathbb F$
    such that the $ab$-entry of $A_i$ is 
    \[ 
        \begin{cases}
        \beta &\text{if $v_a^\beta\in N(a)\cap Q_i$ and $N(b)\cap R_i\neq\emptyset$ for some $\beta$,}\\
        0 & \text{otherwise}.
        \end{cases}
    \] 
    Note that $\beta$ is uniquely chosen because $G_A$ has at most one edge between $N(a)$ and $N(b)$ for~$a\in W_1$ and $b\in W_2$.
    Clearly $A_i$ has rank at most $1$ and $A'=\sum_{i=1}^\ell A_i$.
    Therefore $\rk(A')\le \ell\le  2^k-1$.
    We deduce that $\rk(A[V\cap X_e,V\cap Y_e])\le (2^k-1)+k$. 
    This gives our desired inequality $\rw(A)\leq 2^k+k-1$.
\end{proof}

On the other hand, we prove that 
for a $\sigma$-symmetric matrix $A$
with at most $q-1$ distinct non-zero values of off-diagonal entries, 
if the rank-width of $A$ is $k$,
then the rank-width of~$G_A$ is 
at most $(q-1)(q^k-1)$. 
The key idea is that each row of $A$ across a cut gives rise to at most~$q-1$ rows of the binary adjacency matrix of $G_A$ across the corresponding cut, so bounding the number of distinct rows bounds the cut-rank.

\begin{proposition}[label=thm:matrix_to_graph]
    Let $\mathbb F$ be a field and 
    let $\sigma:\mathbb{F}\rightarrow \mathbb{F}$ be a sesqui-morphism.
    Let $V$ be a finite set and let $A$ be a $V\times V$ $\sigma$-symmetric matrix over $\mathbb{F}$ such that there are at most $q-1$ distinct non-zero values of off-diagonal entries in~$A$.
    The rank-width of the entry graph of~$A$ is at most
    \(  (q-1)(q^{\rw(A)}-1)\).
    Furthermore, 
    we can convert a rank-decomposition of~$A$ of width at most $k$ into  a rank-decomposition of the entry graph of~$A$ of width at most $(q-1)(q^{k}-1)$ in time at most $O(qn)$.
\end{proposition}
\begin{proof}
    If every off-diagonal entry of $A$ is zero, then the entry graph $G_A$ of~$A$ is edgeless and the rank-width of $G_A$ is $0$. So we may assume that  there is at least one non-zero off-diagonal entry and therefore the rank-width of~$A$ is at least $1$ and $\abs{V}\ge 2$. 
    We may assume that every row has non-zero off-diagonal entries because otherwise the corresponding vertex in $G_A$ is isolated and so we can simply delete its row and its corresponding column.
    Let $k=\rw(A)\ge 1$ and let $(T,\mathcal{L})$ be a rank-decomposition of~$A$ of width~$k$.
Since every row has non-zero off-diagonal entries, we have $\abs{W_i}\ge 2$.

    We will construct a rank-decomposition $(T',\mathcal{L}')$ of $G_A$ from $(T,\mathcal{L})$ as follows:
    For each $i\in V$, let $T_i$ be a rooted binary tree whose leaves correspond to the elements of $W_i$.
    Let~$T'$ be a tree obtained from the disjoint union of~$T$ and all~$T_i$'s with $i\in V$ 
    by identifying $\mathcal{L}(i)$ with the root of~$T_i$ for each $i\in V$.
    Let $\mathcal{L}':V(G_A)\rightarrow \{\text{leaves of }T'\}$ 
    be a function such that for each~$i\in V$, 
    $\mathcal L'(i)$ and $\mathcal L'(v_i^\alpha)$ for all $\alpha$
    are in one-to-one correspondence with the set of leaves of~$T_i$. 
    We claim that the width of $(T',\mathcal{L}')$ is at most $(q-1)(q^{k}-1)$.

    Let $e\in E(T')$ and let $(X_e,Y_e)$ be the partition of $V(G_A)$ 
    induced by $\mathcal L'^{-1}$ from the components of $T'\setminus e$.
    If $X_e$ is a proper subset of $W_i$ for some~$i$, 
    then $\rk(B[X_e,Y_e])\le \abs{X_e}< q$, so the width of $e$ is at most $q-1$.
    Thus we may assume that neither $X_e$ nor $Y_e$ is properly contained in~$W_i$ for any $i\in V$
    and therefore $e\in E(T)$. It follows that 
    each of $X_e$ and $Y_e$ is the disjoint union of $W_i$'s.

    Let $U_1=V\cap X_e$ and $U_2=V\cap Y_e$.
    Let $B$ be the adjacency matrix of~$G_A$ over the binary field.
    Then for distinct $i,i'\in V$ and non-zero $\alpha \in\mathbb F$, the $(v_i^\alpha ,v_{i'}^{\sigma(\alpha)})$ entry of $B$ is $1$ if and only if 
    $A_{ii'}=\alpha$.
    So 
    \[ 
    B[W_i,W_{i'}]=\begin{cases}
    \text{$0$-$1$ matrix with exactly one $1$ at the $(v_i^\alpha,v_{i'}^{\sigma(\alpha)})$ entry }
    & 
    \text{if }A_{ii'}=\alpha,\\ 
    \text{zero matrix} & \text{if }A_{ii'}=0.
    \end{cases}
    \] 

    Each row of $A[U_1,U_2]$ corresponding to $i\in U_1$ corresponds to a submatrix $B[W_i,Y_e]$ of~$B[X_e,Y_e]$.
    Since the row of $B[W_i,Y_e]$ indexed by $i$ is zero, at most $q-1$ rows of $B[W_i,Y_e]$ are non-zero.
    
    Furthermore, if two rows of $A[U_1,U_2]$ corresponding to $i,i'\in U_1$ are identical, then $B[W_i,W_j]=B[W_{i'},W_j]$ for each $j\in U_2$, so the set of non-zero rows of $B[W_i,Y_e]$ is the same as the set of non-zero rows of $B[W_{i'},Y_e]$.
Since $\rk (A[U_1,U_2])\le \rw(A)\le k$, $A[U_1,U_2]$ has at most $q^{k}-1$ distinct non-zero rows and therefore we deduce that $B[X_e,Y_e]$ has at most $(q-1)(q^k-1)$ distinct non-zero rows.
    Therefore $\rk(B[X_e,Y_e])\le (q-1)(q^k-1)$, and so the width of $e$ is at most $(q-1)(q^k-1)$.

    Thus the width of $(T',\mathcal{L}')$ is at most $\max\{q-1,(q-1)(q^k-1)\}=(q-1)(q^k-1)$.
    Hence, we conclude that $\rw(G_A)\leq (q-1)(q^k-1)$.
\end{proof}

We remark that Courcelle and Engelfriet~\cite[page 438]{CE2012} presented a transformation that maps an edge-labelled directed graph 
to a simple graph 
with the property that a class of edge-labelled directed graphs has bounded clique-width if and only if the class of its images has bounded clique-width.
As clique-width and rank-width are functionally equivalent~\cite{OS2004,KR2013} and $\sigma$-symmetric matrices can be encoded as edge-labelled directed graphs, their transformation can be seen as an analogue of our entry graph, satisfying qualitative analogues of \zcref{thm:graph_to_matrix,thm:matrix_to_graph}.
The main difference between our entry graph and their transformation is that their transformation converts each vertex into a path, while ours converts each vertex into a star.
More importantly, their proof provides no explicit bounds, as it relies on the existence of a monadic second-order transduction, whereas our propositions give explicit bounds together with efficient translations between the decompositions --- which is what the algorithmic application requires.

By using \zcref{thm:graph_to_matrix,thm:matrix_to_graph} and applying the algorithm of Fomin and Korhonen from \zcref{thm:graph_rankwidth_approximation}, we derive the following algorithm for finding a rank-decomposition of bounded width for a $\sigma$-symmetric matrix of small rank-width with a bounded number of distinct non-zero values of off-diagonal entries.

\begin{proposition}[label=thm:approximating_matrix_rankwidth,store=thm:approximating_matrix_rankwidth]
    Let $n\ge2$.
    Let $k$ be an integer and let $A$ be an $n\times n$ $\sigma$-symmetric matrix over a field~$\mathbb{F}$ such that there are at most $q-1$ distinct non-zero values of off-diagonal entries in $A$.
    Then in time $O_{k,q}(n^2)$ we can either find a rank-decomposition of $A$ of width at most $2^{(q-1)(q^k-1)} +  (q-1)(q^k-1)-1$ 
or correctly conclude that the rank-width of $A$ is more than $k$.
\end{proposition}
\begin{proof}
    For a given $\sigma$-symmetric matrix $A$, we construct the entry graph~$G_A$ of $A$ in time $O(qn^2)$.
    Since $\abs{V(G_A)}\le qn$, 
    by \zcref{thm:graph_rankwidth_approximation} we can either find a rank-decomposition of $G_A$ of width at most $(q-1)(q^{k}-1)$, or conclude that the rank-width of $G_A$ is more than $(q-1)(q^{k}-1)$, in time $O_{k,q}(n^2)$.
    If we find a rank-decomposition of $G_A$, then by \zcref{thm:graph_to_matrix}, we can convert it to a rank-decomposition of $A$ of width at most $2^{(q-1)(q^k-1)} +  (q-1)(q^k-1)-1$ in time $O(qn)$.
    On the other hand, if the rank-width of $G_A$ is more than $(q-1)(q^k-1)$, we conclude that the rank-width of $A$ is more than~$k$ by \zcref{thm:matrix_to_graph}.
\end{proof}

\zcref{thm:approximating_matrix_rankwidth} allows us to obtain the following for matroids given with standard representations.

\begin{theorem}[label=thm:standard_represented_matroid_approx, store=thm:standard_represented_matroid_approx]
    Let $n\ge2$.
    If an $n$-element matroid $M$ is given with its standard representation with at most $q-1$ distinct non-zero values of entries, then for an integer~$k$, 
    in time $O_{k,q}(n^2)$, 
    we can
    either find a branch-decomposition of~$M$ of width at most $2^{(q-1)(q^k-1)}+(q-1)(q^k-1)-1$ or correctly conclude that the branch-width of~$M$ is more than~$k$.
\end{theorem}
\begin{proof}
    By \zcref{lem:matroid_to_symmetric_matrix}, in time $O(n^2)$, we can find a symmetric matrix $A$ such that the rank-width of~$A$ is equal to the branch-width of~$M$ and the number of distinct non-zero values of entries of~$A$ is at most~$q-1$.
    Then we apply the algorithm of \zcref{thm:approximating_matrix_rankwidth} to $A$.
\end{proof}

If a matroid is given with an $m\times n$ matrix which may not be a standard matrix representation, then 
it takes time $O(mn \min(m,n)^{\omega-2})$ to convert it to a standard representation by using \zcref{lem:subdet}.
To limit the number of distinct non-zero values of entries in the standard representation, 
we will limit the number of non-zero subdeterminants.
Thus we deduce the following from \zcref{thm:standard_represented_matroid_approx}.

\begin{corollary}[label=thm:branch-decomposition_subdeterminant, store=thm:branch-decomposition_subdeterminant]
    Let $n\ge2$.
    Let $M$ be a matroid represented by an $m\times n$ matrix~$A$ over a field $\mathbb F$.
    Let~$K$ be a finite subset of $\mathbb F$ such that $-K=K$. 
    Let $q=\abs{K}+1$.
    If every non-zero subdeterminant of~$A$ belongs to~$K$, then for each integer~$k$, in time $O(mn \min(m,n)^{\omega-2}) + O_{k,q}(n^2)$, 
    we can either find a branch-decomposition of~$M$ of width at most $2^{(q-1)(q^k-1)}+(q-1)(q^k-1)-1$ or correctly conclude that the branch-width of $M$ is more than~$k$, where $\omega$ is the matrix multiplication exponent.
\end{corollary}

\section{Finding a rank-decomposition of width at most \texorpdfstring{$k$}{k}}\label{sec:exact}

If the underlying field $\mathbb F$ is finite, we can further improve our algorithm by combining with the dynamic programming algorithm of 
Jeong, Kim, and Oum~\cite{JKO2019}. They proved that one can either find a rank-decomposition of a $\sigma$-symmetric $n\times n$ matrix over $\mathbb F$ of width at most $k$ or confirm that the rank-width is more than $k$ in time $O_{k,\mathbb F}(n^3)$.
They first convert the problem into a problem on the branch-width of \emph{subspace arrangements}
and use the dynamic programming based on the tree structure from a branch-decomposition. For the dynamic programming, they needed an approximate branch-decomposition and they use the technique called the iterative compression to build an approximate branch-decomposition of the input instance. 

We will show that \zcref{thm:approximating_matrix_rankwidth} allows us to skip this iterative compression and provides an approximate branch-decomposition in time $O_{k,\mathbb F}(n^2)$. 
This allows us to obtain a faster $O_{k,\mathbb F}(n^2)$-time algorithm for the problem of finding a branch-decomposition of width at most~$k$ for a matroid with a standard representation over a finite field $\mathbb{F}$.

First, we are going to discuss how we can transform the problem of finding a rank-de\-com\-po\-si\-tion of a $\sigma$-symmetric matrix into the problem of finding a branch-decomposition of a subspace arrangement.

Let $A$ be an $n\times n$ $\sigma$-symmetric matrix over a field $\mathbb{F}$.
Let $v_i$ be the $i$-th column vector of~$A$
and let $\{e_1,e_2,\ldots,e_n\}$ be the standard basis of $\mathbb F^n$. 
Let $V_i$ be the subspace of $\mathbb F^n$ that is spanned by~$e_i$ and~$v_i$ for each $i\in[n]$.
Let $\mathcal V_A$ be the subspace arrangement $\{V_1,V_2,\ldots,V_n\}$.
Note that constructing $\mathcal V_A$ from $A$ can be done in time $O(n^2)$.
The following lemma is due to Kant\'e and Rao~\cite{KR2013}.

\begin{lemma}[label=lem:matrix_to_subspace,note={Kant\'e and Rao~\cite[Corollary 3.34]{KR2013}}]
    Let $A$ be an $n\times n$ $\sigma$-symmetric matrix over a field $\mathbb{F}$.
    Let $f$ be the connectivity function of the subspace arrangement $\mathcal V_A$.
    Then for every subset~$X$ of~$[n]$, 
    $f(X)= 2\cutrk_A(X)$.
    Therefore $\bw(\mathcal{V}_A)=2\rw(A)$.
\end{lemma}

For dynamic programming on a branch-decomposition of a subspace arrangement, it is convenient to have a coordinate system for describing the \emph{boundary space}
$( (\sum_{i\in X}V_i)\cap (\sum_{j\in [n]\setminus X}V_j))$. 
Jeong, Kim, and Oum~\cite{JKO2019} defined the \emph{transcript}, consisting of ordered bases of boundary spaces of every cut appearing in the branch-decomposition. For two consecutive edges in a branch-decomposition, 
one also needs a linear transformation to describe how to change coordinates from one boundary space to the other boundary space. Such a linear transformation is described by \emph{transition matrices}. 
The exact definitions of a transcript and their transition matrices appear in \cite[Subsection 3.5]{JKO2019}.
For our application, it is enough to understand that we can compute these objects as follows. 

\begin{theorem}[label=thm:transcript2,note={Jeong, Kim, and Oum~\cite[Theorem 7.8]{JKO2019}}]
Let $\mathbb F$ be a finite field and $\mathcal V=\{V_1,V_2,\ldots,V_n\}$ a subspace arrangement of subspaces of $\mathbb F^r$ 
    represented by an $r\times m$ matrix $M$ in reduced row echelon form with no zero rows and a partition $\{I_1,I_2,\ldots,I_n\}$  of the column indices such that $V_i$ is the span of the column vectors in~$I_i$ and 
    $\abs{I_i}\le \theta$ for each $i\in\{1,2,\ldots,n\}$.
    Given a rooted branch-decomposition $(T,\mathcal L)$ of $\mathcal V$ of width $\le \theta$, in time $O_{\mathbb F}(\theta^3 n^2)$, we can compute the transcript $\Lambda$ of $(T,\mathcal L)$ with its transition matrices.
\end{theorem}

We observe that 
for a $\sigma$-symmetric matrix $A$,
the representing matrix $(I\mid A)$ for the subspace arrangement $\mathcal V_A$ 
is already in reduced row echelon form, so by using \zcref{thm:transcript2} we obtain the following corollary.

\begin{corollary}[label=thm:transcript]
    Let $A$ be an $n\times n$ $\sigma$-symmetric matrix over a finite field~$\mathbb F$.
    Given a rooted branch-decomposition of~$\mathcal V_A$ of width at most $\theta$, 
    we can compute the transcript~$\Lambda$ with its transition matrices in time $O_{\mathbb F}(\theta^3 n^2)$.
\end{corollary}

Given a rooted branch-decomposition of width at most $\theta$,
the previous theorem provides a transcript with transition matrices.
Now, we are ready to state the algorithm to find a branch-decomposition of width at most $k$ if it exists, when we are given a rooted branch-decomposition of width at most $\theta$ together with a transcript and its transition matrices. In the following theorem, $\mathcal F_x$, called the \emph{full set} at $x$ in \cite{JKO2019}, is a set of compact encodings of partial branch-decompositions used as dynamic-programming table entries.
\begin{theorem}[label=thm:full_set,note={Jeong, Kim, and Oum~\cite[Propositions 6.1, 7.10, and 7.12]{JKO2019}}]
    Let $A$ be an $n\times n$ $\sigma$-symmetric matrix over a finite field~$\mathbb F$.
    Let $k$ be an integer.
    Given a rooted branch-decomposition $(T,\mathcal L)$ of~$\mathcal V_A$ of width at most $\theta$
    with its transcript and transition matrices
    and the root $r$ of $T$, 
    we can compute a set~$\mathcal{F}_x$ at each node $x$ of $T$ in time $O_{k,\theta,\mathbb F}(n)$, satisfying the following.
    \begin{enumerate}[label=\rm(\roman*)]
        \item The branch-width of $\mathcal{V}_A$ is at most $k$ if and only if $\mathcal{F}_r\neq \emptyset$.
        \item If $\mathcal F_r\neq\emptyset$, then we can find a branch-decomposition of $\mathcal V_A$ of width at most $k$ in time 
        $O_{k,\theta,\mathbb F}(n)$.
\end{enumerate}
\end{theorem}

The following theorem is the restatement of the result of Jeong, Kim, and Oum~\cite{JKO2019} for our purposes by combining the previous two theorems.

\begin{theorem}[label=thm:branch-width_compression,note={Jeong, Kim, and Oum~\cite{JKO2019}}]
    Let $n\ge2$. Let $k$, $\theta$ be integers.
    Let $A$ be an $n\times n$ $\sigma$-symmetric matrix over a finite field~$\mathbb F$.
    Let $(T,\mathcal L)$ be a rank-decomposition of $A$ of width at most $\theta$.
    Then, in time
$O_{k,\theta,\mathbb F}(n^2)$,
    we can either find a rank-decomposition of~$A$ of width at most $k$ or correctly conclude that the rank-width of~$A$ is more than $k$.
\end{theorem}
\begin{proof}
    We may assume that $k>0$ because otherwise it is trivial. 
    By \zcref{lem:matrix_to_subspace}, we can construct a subspace arrangement $\mathcal V_A$ of $\mathbb{F}^n$ with $\bw(\mathcal{V}_A)=2\rw(A)$ in time $O(n^2)$.
    We can regard $(T,\mathcal{L})$ as a branch-decomposition of $\mathcal{V}_A$ of width at most~$2\theta$ and pick a root of $T$ by subdividing an edge to create a root to make it a rooted branch-decomposition.
    We run the algorithm in~\zcref{thm:transcript} to compute a transcript and its transition matrices in time $O((2\theta)^3 n^2)$.
    The algorithm in \zcref{thm:full_set} runs in time 
    $O_{k,\theta,\mathbb F}(n)$
to find a branch-decomposition of $\mathcal V_A$ of width at most $2k$ or correctly conclude that the branch-width of $\mathcal V_A$ is more than~$2k$.
    Now observe that a branch-decomposition of~$\mathcal V_A$ of width at most~$2k$ can be regarded as a rank-decomposition of~$A$ of width at most~$k$ by \zcref{lem:matrix_to_subspace}.
\end{proof}
Now, let us combine with \zcref{thm:approximating_matrix_rankwidth}  to deduce the following. 

\begin{theorem}[label=thm:exact_algorithm, store=thm:exact_algorithm]
    Let $n\ge2$. Let $\mathbb F$ be a finite field. 
    Given an $n\times n$ $\sigma$-symmetric matrix~$A$ over~$\mathbb{F}$, 
    for an integer~$k$, 
    we can either find a rank-decomposition of $A$ of width at most $k$ or correctly conclude that the rank-width of $A$ is more than $k$ in time $O_{k,\mathbb F}(n^2)$.
\end{theorem}
\begin{proof}
    We may assume that $k>0$ because otherwise it is trivial. 
    Let $q=\abs{\mathbb F}$.
    First, we apply \zcref{thm:approximating_matrix_rankwidth} to either find a rank-decomposition of $A$ that has width at most $\theta= 2^{(q-1)(q^k-1)}+(q-1)(q^k-1)-1$ or conclude that the rank-width of~$A$ is more than~$k$ in time $O_{k,q}(n^2)$.
    Suppose that we obtained a rank-decomposition $(T,\mathcal{L})$ of~$A$ whose width is at most $\theta$.
Now we apply \zcref{thm:branch-width_compression} to either find a rank-decomposition of~$A$ of width at most~$k$ or correctly conclude that the rank-width of $A$ is more than $k$ in 
    time~$O_{k,\theta,q}(n^2)$.
So, the total running time is at most $O_{k,\mathbb F}(n^2)$.
\end{proof}

The current proof method does not allow us to generalize \zcref{thm:exact_algorithm} to an infinite field 
because the proof of Jeong, Kim, and Oum~\cite{JKO2019} depends on the finiteness of the number of subspaces of a vector space over $\mathbb{F}$.

For matroids given with a standard representation, we obtain the following.

\begin{theorem}[label=thm:standard_represented_matroid_branchwidth_approximation, store=thm:standard_represented_matroid_branchwidth_approximation]
    Let $n\ge2$.
    For each integer $k$ and an $n$-element matroid $M$ given by a standard representation over a finite field~$\mathbb F$, in time $O_{k,\mathbb F}(n^2)$ we can either find a branch-decomposition of $M$ of width at most~$k$ or correctly conclude that the branch-width of $M$ is more than~$k$.
\end{theorem}
\begin{proof}
    By \zcref{lem:matroid_to_symmetric_matrix}, in time $O(n^2)$, we can find a symmetric matrix $A$ whose rank-width is equal to the branch-width of $M$.
    Then we apply the algorithm of \zcref{thm:exact_algorithm} to $A$.

    We remark that alternatively one could use \zcref{thm:standard_represented_matroid_approx} to obtain a branch-decomposition of bounded width  
    and then compute its transcript with transition matrices by \zcref{thm:transcript2} and finally use \zcref{thm:branch-width_compression}. 
    The running time is still $O_{k,\mathbb F}(n^2)$.
\end{proof}

As before, if the representation of the given matroid is not in standard form, we can use \zcref{lem:subdet} to obtain the following corollary from \zcref{thm:standard_represented_matroid_branchwidth_approximation}.

\begin{corollary}[label=thm:represented_matroid_branchwidth_approximation, store=thm:represented_matroid_branchwidth_approximation]
    Let $n\ge2$.
    Let $M$ be a matroid given by an $m\times n$ matrix representation over a finite field~$\mathbb{F}$.
    For an integer~$k$, in time $O(mn \min(m,n)^{\omega-2})+O_{k,\mathbb F}(n^2)$
    we can either find a branch-decomposition of $M$ of width at most $k$ or correctly conclude that the branch-width of~$M$ is more than $k$, where $\omega$ is the matrix multiplication exponent.
\end{corollary}

Kant\'e and Rao~\cite{KR2013} defined the \emph{rank-width} of a directed graph 
as follows. 
Let $\{0,1,\alpha,\alpha^2\}$ be the elements of $\mathbb F_4$, the finite field with four elements, 
such that $1+\alpha+\alpha^2=0$, $\alpha^3=1$, and the characteristic of $\mathbb F_4$ is $2$.
Let $\sigma$ be an automorphism on $\mathbb{F}_4$ 
such that $\sigma(0)=0$, $\sigma(1)=1$, $\sigma(\alpha)=\alpha^2$ and $\sigma(\alpha^2)=\alpha$. Clearly, $\sigma$ is a sesqui-morphism.
For a directed graph $G$, let us define its $\mathbb F_4$-adjacency matrix $A_G$ as a $V(G)\times V(G)$ matrix such that 
the $ij$-entry of $A_G$ is 
\[ 
\begin{cases}
    1 & \text{if }(i,j)\in E(G)\text{ and }(j,i)\in E(G), \\ 
    \alpha &\text{if }(i,j)\in E(G)\text{ and } (j,i)\notin E(G),\\ 
    \alpha^2 &\text{if }(i,j)\notin E(G)\text{ and } (j,i)\in E(G),\\ 
    0 &\text{if }(i,j)\notin E(G)\text{ and } (j,i)\notin E(G).
\end{cases}
\] 
Now, this matrix $A_G$ is clearly $\sigma$-symmetric. 
The \emph{rank-width} of a directed graph $G$, denoted by $\rw^{\mathbb F_4}(G)$, is defined as the rank-width of $A_G$ \cite{KR2013}.

As a corollary of \zcref{thm:exact_algorithm}, we deduce the following $O_{k}(n^2)$-time algorithm. 
The previous algorithm due to Kant\'e and Rao~\cite{KR2013} runs in time $O_k(n^3)$.

\begin{corollary}[label=cor:directedrankwidth,store=cor:directedrankwidth]
    Let $n\ge2$.
    For an integer~$k$, in time $O_k(n^2)$, we can find a rank-decomposition of an input $n$-vertex  directed graph~$G$ of width at most $k$ or confirm that the rank-width of $G$ is more than $k$.
\end{corollary}

\section{Finding a path-decomposition quickly}\label{sec:linear_rank-width}

Jeong, Kim, and Oum~\cite{JKO2016} showed that we can decide whether a $\sigma$-symmetric $n\times n$ matrix over a fixed finite field 
has linear rank-width at most~$k$ for each fixed $k$ in time $O_{k,\mathbb F}(n^3)$.
We want to show that as an application of \zcref{thm:approximating_matrix_rankwidth}, 
we can improve $O_{k,\mathbb F}(n^3)$ to $O_{k,\mathbb F}(n^2)$ and this will allow us to obtain such an algorithm also for the path-width of matroids by using the reduction in \zcref{sec:matroidtomatrix}.

Let us first recall the algorithm of Jeong, Kim, and Oum~\cite{JKO2016}
to find a linear layout of width at most $k$ based on dynamic programming, assuming that the input has small branch-width and is given with a branch-decomposition of small width
and a transcript with its transition matrices. 
\begin{theorem}[label=thm:art_of_trellis,note={Jeong, Kim, and Oum~\cite[Propositions 43 and 44]{JKO2016}}]
    Let $k\ge0$, $\theta\ge 1$ be integers. Let $\mathbb F$ be a finite field. 
    Let $\mathcal{V}=\{V_1,V_2,\ldots,V_n\}$ be a subspace arrangement of subspaces of~$\mathbb{F}^r$.
    Let~$(T,\mathcal{L})$ be a given rooted branch-decomposition of $\mathcal{V}$ of width at most $\theta$.
    If we are given the transcript~$\Lambda$ of~$(T,\mathcal L)$ with its transition matrices, 
    then one can find a linear layout of width at most $k$ or confirm that no such linear layout exists in time 
$O_{\theta,\mathbb F,k}(n)$.
\end{theorem}
In the paper of Jeong, Kim, and Oum~\cite[Proposition 39]{JKO2016} for path-width,
the transcript was computed by a slower, cubic-time algorithm.  
Following the thesis of Jeong~\cite{Jeong2018}, 
we instead use \zcref{thm:transcript2} 
from their later paper~\cite{JKO2019}.

By \zcref{lem:matrix_to_subspace}, we can reduce the problem of identifying $\sigma$-symmetric matrices of linear rank-width at most $k$ 
to a problem of identifying subspace arrangements of path-width at most~$2k$.
Thus, by \zcref{thm:approximating_matrix_rankwidth,thm:art_of_trellis}, we obtain a quadratic-time algorithm for deciding whether the linear rank-width of a $\sigma$-symmetric matrix over a fixed finite field is at most~$k$ for fixed~$k$.

\begin{theorem}[label=thm:exact_linear_rank-width]
    Let $\mathbb{F}$ be a finite field.
    Let $k$ be an integer.
    Given an $n\times n$ $\sigma$-symmetric matrix~$A$ over~$\mathbb{F}$, 
    we can either find a linear rank-decomposition of $A$ of width at most~$k$ or correctly conclude that the linear rank-width of $A$ is more than~$k$ in time $O_{k,\mathbb F}(n^2)$.
\end{theorem}
\begin{proof}
    We may assume that $k>0$ because otherwise it is trivial.
We first use \zcref{thm:exact_algorithm} to either conclude that $\rw(A) > k$ or find a rank-decomposition of $A$ of width at most $k$.
If $\rw(A)>k$, then by \zcref{lem:lrwleqrw}, we conclude that $\lrw(A)>k$.
Thus, assume that we find a rank-decomposition $(T,\mathcal L)$ of $A$ of width at most~$k$.
    By picking an edge and subdividing it to create a root, we may assume that $(T,\mathcal L)$ is a rooted rank-decomposition.
    Let \[ M= ( I \mid A )\]  be an $n\times 2n$ matrix over $\mathbb F$ where $I$ is the $n\times n$ identity matrix.
    Assume that $\{1,2,\ldots,2n\}$ is the set of indices of columns of $M$. 
    Let $I_{i}=\{i,n+i\}$
    and $V_i$ be the span of the column vectors indexed by elements in $I_i$ for all $i\in \{1,2,\ldots,n\}$.
    Then $\mathcal V=\{V_1,V_2,\ldots,V_n\}$ is a subspace arrangement such that its path-width is equal to twice the linear rank-width of $A$ 
    and $(T,\mathcal L)$ is a rooted branch-decomposition of width at most $\theta:=2k$ of $\mathcal V$
    by \zcref{lem:matrix_to_subspace}.
    Since $M$ is in reduced row echelon form with no zero rows, 
    we apply \zcref{thm:transcript2} to obtain the transcript and its transition matrices for $(T,\mathcal L)$ in time $O_{k,\mathbb F}(n^2)$.
    With \zcref{thm:art_of_trellis}, we can either find 
    a linear layout of $\mathcal V$ of width at most $2k$ 
    or conclude that no such linear layout exists
    in time $O_{k,\mathbb F}(n)$. 
    Linear layouts of $\mathcal V$ of width at most $2k$ 
    are precisely linear rank-decompositions of $A$ of width at most $k$. 
    Thus the proof is complete and the running time is $O_{k,\mathbb F}(n^2)$.
\end{proof}

Since the linear rank-width of a graph is equal to the linear rank-width of its adjacency matrix over the binary field, we obtain the following.
This corollary can also be obtained as a corollary of the theorem of Fomin and Korhonen (\zcref{thm:graph_rankwidth_approximation}) with \zcref{thm:transcript2,thm:art_of_trellis}.

\begin{corollary}[label=thm:linear_rank-width, store=thm:linear_rank-width]
    Let $k$ be an integer.
    We can either find a linear rank-decomposition of an input $n$-vertex graph $G$ of width at most~$k$ or correctly conclude that the linear rank-width of $G$ is more than~$k$ in time $O_k(n^2)$.
\end{corollary}
\begin{proof}
    We apply \zcref{thm:exact_linear_rank-width} to the adjacency matrix of~$G$.
\end{proof}

For matroids given with a standard representation, we have the following.

\begin{corollary}[label=thm:exact_path-width_b,store=thm:exact_path-width_b]
    Let $n\ge2$ and $k$ be integers.
    For an $n$-element matroid $M$ given by a standard representation over a finite field~$\mathbb F$, in time $O_{k,\mathbb F}(n^2)$ we can either find a path-decomposition of $M$ of width at most~$k$ or correctly conclude that the path-width of $M$ is more than~$k$.
\end{corollary}
\begin{proof}
    By \zcref{lem:matroid_to_symmetric_matrix}, in time $O(n^2)$, we can find a symmetric matrix $A$ whose linear rank-width is equal to the path-width of $M$. 
    Then we apply the algorithm of \zcref{thm:exact_linear_rank-width} to~$A$.

    We remark that alternatively one could use \zcref{thm:standard_represented_matroid_approx} to obtain a branch-decomposition of bounded width  
    and then compute its transcript with transition matrices by \zcref{thm:transcript2} and finally use \zcref{thm:art_of_trellis}. 
    The running time is still $O_{k,\mathbb F}(n^2)$.
\end{proof}

As before, if the input matroid is given with a general representation, 
we can use \zcref{lem:subdet} to obtain the following corollary from \zcref{thm:exact_path-width_b}.

\begin{corollary}[label=thm:path-width, store=thm:path-width]
Let $k$ be an integer and let $M$ be a matroid given by an $m\times n$ matrix representation over a finite field $\mathbb{F}$.
    Then in time $O(mn \min(m,n)^{\omega-2})+O_{k,\mathbb F}(n^2)$ we can either find a path-decomposition of~$M$ of width at most~$k$ or correctly conclude that the path-width of~$M$ is more than~$k$, where $\omega$ is the matrix multiplication exponent.
\end{corollary}

\section{Reduction to matrix singularity}\label{sec:reduction}
We give evidence that the factor of $n^{\omega}$ is unavoidable for computing matroid branch-width when the matroid is given by a general representation.
First, we make the simple observation that 
connects the problem of deciding whether the branch-width of a matroid is zero 
with the problem of deciding whether the column vectors of a matrix are linearly independent.

\begin{lemma}[label=lem:branchwidth0, store=lem:branchwidth0]
    Let $A$ be a matrix over a field $\mathbb{F}$ without a zero column.
    Then $\bw(M(A))=0$ if and only if the columns of $A$ are linearly independent.
\end{lemma}
\begin{proof}
    Let $M=M(A)$ and $E=E(M)$.
    Suppose that the columns of $A$ are linearly independent.
    Then for any subset $X\subseteq E$, we have $\rk_M(X)=\abs{X}$.
    Therefore, the connectivity function~$\lambda_M$ is always $0$, which implies $\bw(M)=0$.

    Next, suppose that the columns of $A$ are linearly dependent.
    Then we can find a column $v$ of~$A$ that can be written as a linear combination of other columns of $A$.
    Therefore if we delete the column $v$ from $A$, the rank does not change.
    Let $(T,\mathcal{L})$ be a branch-decomposition of $M$.
    Then the width of the unique edge of $T$ incident with the leaf corresponding to~$v$ is
    \[\lambda_M(\{v\})=\rk_M(\{v\})+\rk_M(E\setminus\{v\})-\rk_M(E)=1,\]
    so the width of $(T,\mathcal{L})$ is at least $1$.
    Therefore, we have $\bw(M)\geq 1$.
\end{proof}

More generally, we aim to show that for every fixed $k$, deciding whether the branch-width is at most $1.5k$ or is more than $k$ for a matroid represented over a fixed field $\mathbb F$ with at least $4$ elements is as hard as the problem of deciding whether a square matrix over $\mathbb F$ is non-singular.
We give a quadratic-time reduction from \textsc{Matrix Singularity} to \textsc{Branch-width $1.5$-Approximation for~$k$}.
So, if for a fixed field $\mathbb F$ there is no algorithm that runs in time $O(n^{\omega-\varepsilon})$
to decide an $n\times n$ matrix is non-singular for some $0<\varepsilon\le \omega-2$,
then for each fixed~$k$,
it is impossible to decide whether the branch-width of a given represented matroid over $\mathbb F$ is at most $1.5k$ or more than $k$ 
in time $O(n^{\omega-\varepsilon})$.
Furthermore, if the given field is sufficiently large, then we will show that $1.5k$ can be improved to $2k$.

We start by formally defining two algorithmic problems we will cover.
We assume that $\mathbb F$ is a fixed field.

\begin{tcolorbox}\textsc{Matrix Singularity}

\textbf{Input:} an $n\times n$ matrix $A$ over $\mathbb{F}$.

\textbf{Output:} is $A$ non-singular?
\end{tcolorbox}

Let $k$ be a fixed non-negative integer and let $r\geq 1$ be a fixed real number.

\begin{tcolorbox}\textsc{Branch-width $r$-Approximation for $k$}

\textbf{Input:} a matrix $A$ over $\mathbb{F}$.

\textbf{Output:} For a matroid $M$ represented by $A$ over $\mathbb{F}$, confirm $\bw(M)\leq r\cdot k$ or confirm\\ ${}$ \hspace{1.5cm} $\bw(M)> k$.
\end{tcolorbox}

The main strategy of the reduction is that for an $n\times n$ matrix $A$, we construct a new matrix such that the matroid represented by the new matrix has branch-width determined by the determinant of $A$.
Our starting point is \zcref{lem:branchwidth0}, which connects the branch-width of a represented matroid to the singularity of a matrix.
We will introduce three operations on matrix representations that manipulate the branch-width in a controlled way.
These operations will be our key method of the reduction.

\subsection{Titanic and robust partitions}
To move forward to the next step, we introduce the titanic partition, originating from the definition of a titanic separation due to Robertson and Seymour~\cite{RS1991}.
Let $\lambda$ be a connectivity function 
on the subsets of a finite set~$E$.
A subset $X$ of $E$ is called \emph{titanic} with respect to $\lambda$ if whenever $X_1$, $X_2$, $X_3$ are pairwise disjoint subsets of $X$ such that $X_1\cup X_2\cup X_3=X$, there is~$i\in \{1,2,3\}$ such that $\lambda(X_i)\geq \lambda(X)$.
A partition $\mathcal{P}$ of $E$ is called \emph{titanic} with respect to $\lambda$ if every part of $\mathcal{P}$ is titanic with respect to $\lambda$.
The following lemma shows that for a matroid, if each part of a titanic partition is small under the connectivity function, then the branch-width of the partitioned matroid is also small.

\begin{lemma}[label=lem:titanic_partition,note={Hlin\v{e}n\'y and Oum~\cite[Lemma 3.4]{HO2006}}]
    Let $\lambda$ be a connectivity function on the subsets of a finite set~$E$ 
    and let $\mathcal P$ be a titanic partition of $E$.
    Let $\lambda^{\mathcal P}$ be a connectivity function on the subsets of~$\mathcal P$ such that $\lambda^{\mathcal P}(U)=\lambda(\bigcup_{X\in U}X)$ for all $U\subseteq \mathcal P$.
    If $\bw(\lambda) \le k$ and $\lambda(X)\le k$ for all $X\in \mathcal P$, then 
    the branch-width of $\lambda^{\mathcal P}$ is at most $k$.
\end{lemma}
In particular, this implies that if $M=(E,\mathcal{I})$ is a matroid of branch-width at most $k$
and~$\mathcal P$ is a titanic partition on $E$ such that $\lambda_M(X)\leq k$ for all $X\in \mathcal P$, 
then the branch-width of the partitioned matroid~$(M,\mathcal{P})$ is at most $k$. This is a useful tool to find an optimal branch-decomposition using our preferred cuts.

Now let us introduce a robust partition.
A subset $X$ of $E$ is \emph{robust} with respect to~$\lambda$ if 
for all partitions $(X_1,X_2)$ of $X$, we have $\lambda(X_1)\ge \lambda(X)$ or $\lambda(X_2)\ge \lambda(X)$.
A partition $\mathcal{P}$ of $E$ is called \emph{robust} with respect to $\lambda$ if every part of $\mathcal{P}$ is robust with respect to $\lambda$.

\begin{lemma}[label=lem:robust_partition]
    Let $\lambda $ be a connectivity function on subsets of a finite set $E$.
    Let $\mathcal P$ be a robust partition of~$E$ such that
    \begin{enumerate}[label=\rm(\roman*)]
\item\label{itm:not1} $\lambda(\bigcup_{X\in \mathcal Q} X)\neq1$ for all~$\mathcal Q\subseteq \mathcal P$,
        and     
        \item\label{itm:noncrossing} for all $X\subseteq E$, if there are at least two parts $X_1,X_2\in \mathcal P$ such that $X\cap X_i\neq \emptyset$ and $X\cap X_i\neq X_i$ for $i\in\{1,2\}$, then  $\lambda(X)>1$.
    \end{enumerate}
If $\lambda$ has a branch-decomposition of width at most $1$, then 
    $\lambda(X)=0$ for all $X\in \mathcal P$.
\end{lemma}

\begin{proof}
    Let $(T,\mathcal L)$ be a branch-decomposition of $\lambda$ of width at most $1$. 
    We may assume that $T$ has no node of degree~$2$.
    If $\mathcal P=\{E\}$, then trivially $\lambda(E)=0$. 
    Thus we may assume that $\abs{\mathcal P}>1$.
    If~$\abs{E}=2$, then by \ref{itm:not1} and the fact that the branch-width of $\lambda$ is at most $1$, we deduce easily that~$\lambda(X)=0$ for all $X\in\mathcal P$.
    Thus we may assume that $\abs{E}\ge 3$.

    For each $e\in E(T)$ and a node $v$ of~$T$, let $X_{e,v}$ be the set of elements of $E$ that are mapped by $\mathcal{L}$ to nodes in the component of $T-e$ not containing $v$.
    We say that a part $X\in \mathcal{P}$ is \emph{good} if~$\lambda(X)=0$ and \emph{bad} otherwise.
    Suppose that there is at least one bad part in $\mathcal{P}$.

    First, suppose that $X\in \mathcal{P}$ is a part such that $X=X_{e,v}$ for some $e\in E(T)$ and an end~$v$ of~$e$.
    The width of $e$ is $\lambda(X)\neq 1$ by \ref{itm:not1}, while the width of~$e$ is at most~$1$ since the width of $(T,\mathcal{L})$ is at most~$1$.
    Therefore, we deduce that $\lambda(X)=0$ so $X$ is a good part.

    Now choose an arbitrary non-leaf node $r\in V(T)$ to be a root of $T$.
    We regard $T$ as a rooted tree with the root $r$ and we will use the terminologies such as ancestors and descendants with respect to the root $r$.
    Note that since $\abs{E}\geq 3$, $T$ has at least one non-leaf node.
For each part~$X\in \mathcal{P}$, let $t_X\in V(T)$ be the lowest common ancestor of $\mathcal{L}(X)$.
    Let $A\in \mathcal{P}$ be a bad part such that the distance between $t_A$ and $r$ is maximized among all bad parts.
Note $\abs{A}\ge2$, since every singleton part equals some $X_{e,v}$ and is therefore good.
    
    Next we define an orientation on $E(T)$ as in the proof of \cite[Proposition 2.4]{KO2026}.
    Let us orient an edge $e=uv$ of~$T$ toward $v$ if $\lambda(X_{e,v}\setminus A)\leq \lambda(X_{e,v})$ and $A\setminus X_{e,v}\neq\emptyset$.
    Since $A$ is robust, we have~$\lambda(X_{e,u}\cap A)\geq \lambda(A)$ or $\lambda(X_{e,v}\cap A)\geq \lambda(A)$.
    If $\lambda(X_{e,u}\cap A)\geq \lambda(A)$, then by the submodularity,
    \[\lambda(X_{e,v})+\lambda(A)=\lambda(X_{e,u})+\lambda(A)\geq \lambda(X_{e,u}\cap A)+\lambda(X_{e,u}\cup A)\geq \lambda(A)+\lambda(X_{e,v}\setminus A),\]
    so we have $ \lambda(X_{e,v}\setminus A)\leq \lambda(X_{e,v})$.
    Similarly, if $\lambda(X_{e,v}\cap A)\geq \lambda(A)$, we have $\lambda(X_{e,u}\setminus A)\leq \lambda(X_{e,u})$.
    Observe that if $A\setminus X_{e,v}=\emptyset$, then $e$ is oriented toward $u$ but is not oriented toward $v$.
    Thus every edge of $T$ receives at least one orientation.
    Note that an edge $e$ can be oriented in both directions.

For each $v\in E$, since $\lambda(\{v\}\setminus A)\leq \lambda(\{v\})$, each edge incident with a leaf is oriented away from the leaf.
    So, we assume that every edge incident with a leaf is not oriented toward the leaf by removing such an orientation of the edge.

    By the pigeonhole principle, $T$ has a node $s$ such that every edge incident with $s$ is oriented toward~$s$. This can be seen from the fact that in a tree obtained from $T$ by contracting all edges directed in both ways, the number of edges is less than the number of nodes.
    Since every edge incident with a leaf is not oriented toward the leaf, 
    $s$ is not a leaf and so $s$ has degree~$3$.

    Let $e_1,e_2,e_3$ be the edges of $T$ incident with $s$.
    We may assume that if $r\neq s$,
    then the path from~$r$ to~$s$ on~$T$ contains $e_1$.
    Let $v_1,v_2,v_3$ be the endpoints of $e_1,e_2,e_3$ other than $s$, respectively.
    Let $X_i=X_{e_i,s}$ for each $i\in\{1,2,3\}$.
    Since $e_i$ is oriented toward $s$, we have 
    \[
    \lambda(X_i\setminus A)\le \lambda(X_i)\text{ and } A\setminus X_i\neq\emptyset. 
    \] 
    Thus at least two of $X_1\cap A$, $X_2\cap A$, and $X_3\cap A$ are nonempty.

    If $B\in \mathcal P$ has the property that $B\neq A$ and 
    at least two of $X_1\cap B$, $X_2\cap B$, and $X_3\cap B$ are nonempty,
    then there exists $i\in\{1,2,3\}$ such that 
    all of $X_i\cap A$, $X_i\cap B$, $(E\setminus X_i)\cap A$, and $(E\setminus X_i)\cap B$ are nonempty.
    This contradicts the assumption that $\lambda(X_i)\le 1$ and \ref{itm:noncrossing}.
    Therefore for each $B\in \mathcal P$ with $B\neq A$, 
    we have $B\subseteq X_i$ for some $i\in\{1,2,3\}$.

    As we chose $A$ to maximize the distance from the root to $t_A$,
    every part $B\in\mathcal P$ in $X_2$ or $X_3$ is good
    and therefore $\lambda(X_2\setminus A)=0$ and $\lambda(X_3\setminus A)=0$
    by the submodular inequality
    because each of $X_2\setminus A$ and $X_3\setminus A$ is the disjoint union of sets whose $\lambda$-values are zero.
    Note that $\lambda(X_1\setminus A)\le\lambda(X_1)\le1$.
    Now, $\lambda(A)=\lambda(E\setminus A)\le \lambda(X_1\setminus A)+\lambda(X_2\setminus A)+\lambda(X_3\setminus A)\le 1$.
    By \ref{itm:not1}, $\lambda(A)=0$, contradicting the assumption that $A$ is bad.
\end{proof}

\subsection{Constructing a new matroid of branch-width \texorpdfstring{$1$}{1} or larger}

Our first operation is taking the tensor product with $\left(\begin{smallmatrix} 1 & 0 & 1 \\ 0 & 1 & 1 \end{smallmatrix}\right)$.
Applying this operation to a representation of a matroid, we can construct a new represented matroid of branch-width at least~$1$.
We first provide an upper bound of the branch-width of $M\left(\left(\begin{smallmatrix} 1 & 0 & 1 \\ 0 & 1 & 1 \end{smallmatrix}\right)\otimes A\right)$.

\begin{lemma}\label{lem:tensor}
    Let $A$, $B$ be matrices over a field $\mathbb F$.
    Then 
    the branch-width of $M(B\otimes A)$ is at most 
    $\max(\rk(B) \bw(M(A)), \bw(M(B)))$.
\end{lemma}
\begin{proof}
    Let $E_A$, $E_B$ be the set of column indices of $A$ and $B$, respectively.
    Let $M=M(A)$ and $N=M(B\otimes A)$. 
    Note that $E(M)=E_A$ and $E(N)=E_B\times E_A$.
    By \zcref{lem:tensor_rank}, $\lambda_N(E_B\times X)=\rk(B)\lambda_M(X)$ for all $X\subseteq E_A$.

    If $\abs{E_A}\le 1$, then it is easy to see that 
    $\bw(N)\le \bw(M(B))$. 
    Thus we may assume that $\abs{E_A}>1$. 
    If $\abs{E_B}\le 1$, then trivially $\bw(N)\le \rk(B)\bw(M(A))$.
    Hence we may also assume that $\abs{E_B}>1$.
    Let $(S,\mathcal{K})$ be a branch-decomposition of $M$ witnessing $\bw(M)=a$.
    We construct a branch-decomposition~$(T,\mathcal{L})$ of $N$ as follows.
    Let $T$ be a subcubic tree obtained from the disjoint union of~$S$ and $\abs{E_A}$ copies of rooted branch-decompositions of $M(B)$ of width $\bw(M(B))$
    by identifying each leaf $\ell=\mathcal{K}(v)$ for~$v\in E_A$ with the root of a rooted branch-decomposition.
    We define $\mathcal{L}$ so that for each~$v\in E_A$, elements in $E_B\times \{v\}$ are mapped by~$\mathcal{L}$ to distinct leaves of a rooted branch-decomposition whose root is identified with $\mathcal{K}(v)$,
    where this mapping should respect the mapping in the branch-decomposition of $M(B)$ from $E_B$ to the leaves.

    We claim that $(T,\mathcal L)$ has width at most $\max\{\rk(B)a,\bw(M(B))\}$. 
    Let $e\in E(T)$ and let $(X_e,Y_e)$ be the partition of $E_B\times E_A$ induced by $\mathcal{L}$ and the connected components of $T\setminus e$.
    It is enough to show that~$\lambda_N(X_e)\le \max\{\rk(B)a,\bw(M(B))\}$. 
    If $X_e=E_B\times X$ for some $X\subseteq E_A$, then 
    by \zcref{lem:tensor_rank},
    $\lambda_N(X_e)=\rk(B)\lambda_M(X)$ because~$(X,E_A\setminus X)$ is a partition induced by $S\setminus e$.
    Thus we may assume that there is some $v\in E_A$ such that $X_e\cap (E_B\times \{v\})\neq \emptyset$ and $Y_e\cap (E_B\times \{v\})\neq \emptyset$.

    Then by the construction of $T$, either $X_e\subseteq (E_B\times \{v\})$ or $Y_e\subseteq (E_B\times \{v\})$.
    Without loss of generality, assume that $X_e\subseteq (E_B\times \{v\})$.
    Let $X_e'$ be a subset of $E_B$ such that $X_e=X_e'\times\{v\}$.
    Let~$Y_e'=E_B\setminus X_e'$.
    Then \[ \lambda_N(X_e)\le \rk_N(X_e)=\rk_M (\{v\})\rk_{M(B)}(X_e')\le \rk_{M(B)} (X_e')\le \rk(B)\] by \zcref{lem:tensor_rank}.
    If $a\ge 1$, then $\lambda_N(X_e)\le \rk(B)\le \rk(B)a$ as desired.
    So we may assume that $a=0$ and therefore $\lambda_M(\{v\})=0$.
    Then by the submodular inequality and \zcref{lem:tensor_rank}, we have 
    \begin{multline*}
        \lambda_N(X_e)=\lambda_N(Y_e)\leq \lambda_N(E_B\times (E_A\setminus\{v\}))+\lambda_N(Y_e'\times\{v\})\\
        = \rk(B)\lambda_M(\{v\})+ \rk_M(\{v\})\lambda_{M(B)}(Y_e')
        \le \lambda_{M(B)}(Y_e').
    \end{multline*}
    Since $(X_e',Y_e')$ is a partition arising from a branch-decomposition of $M(B)$ of width $\bw(M(B))$, we deduce that $\lambda_{M(B)}(Y_e')\le \bw(M(B))$. Therefore $\lambda_N(X_e)\le \bw(M(B))$ if $a=0$.
    This proves that the branch-width of $N$ is at most $\max\{\rk(B)a,\bw(M(B))\}$.
\end{proof}

Here is our key proposition for the first reduction $\left(\begin{smallmatrix} 1 & 0 & 1 \\ 0 & 1 & 1 \end{smallmatrix}\right)\otimes A$.
\begin{proposition}[label=thm:branchwidth_0_1]
    Let $M$ be a simple matroid on $E$ represented by a matrix $A$ over a field $\mathbb F$.
    Let $N=(E',\mathcal{I}')$ be a matroid represented by $\left(\begin{smallmatrix} 1 & 0 & 1 \\ 0 & 1 & 1 \end{smallmatrix}\right)\otimes A$.
    Then we have the following. 
    \begin{itemize}
        \item $\bw(N)>1$ if $\bw(M)>0$.
        \item $\bw(N)=1$ if $\bw(M)=0$.
\end{itemize}
\end{proposition}
\begin{proof}
    Let $F$ be the set of column indices of $\left(\begin{smallmatrix} 1 & 0 & 1 \\ 0 & 1 & 1 \end{smallmatrix}\right)$.
    Then $E'=F\times E$.
    For each $v\in E$, let~$E_v=F\times\{v\}$.

    Note that $A$ is non-zero since $M$ is a simple matroid.
Observe that $N$ contains a circuit of size~$3$, and the restriction of $N$ to a circuit of size~$3$
has branch-width $1$.
    Therefore $\bw(N)\ge 1$.

    By applying \zcref{lem:tensor} with $B:=\left(\begin{smallmatrix} 1 & 0 & 1 \\ 0 & 1 & 1 \end{smallmatrix}\right)$, 
    we deduce that 
    if $\bw(M)=0$, then $\bw(N)= 1$.

    Now, suppose that $\bw(M)>0$.
    To obtain a contradiction, suppose that $\bw(N)\leq 1$ and let~$(T,\mathcal{L})$ be a branch-decomposition of $N$ of width at most~$1$.

    First observe that for each $v\in E$, any two elements in $E_v$ spans $E_v$.
    Thus if $(X_1,X_2)$ is a partition of $E_v$, then we have $\lambda_N(X_1)=\lambda_N(E_v)$ or $\lambda_N(X_2)=\lambda_N(E_v)$.
    Therefore, $\mathcal P=\{E_v\colon v\in E\}$ is a robust partition of $E'$.

    For all $Y\subseteq E$,
    we have $\lambda_N(F\times Y)=2\lambda_M(Y)\neq 1$
    by \zcref{lem:tensor_rank}.

    Moreover, suppose that there is $X\subseteq E'$ and two distinct elements $u,v\in E$ such that $X\cap E_w\neq \emptyset$ and  $X\cap E_w\neq E_w$ for each $w\in \{u,v\}$.
    Let $X_v=X\cap E_v$, $X_u=X\cap E_u$, $Y_v=E_v\setminus X$, and $Y_u=E_u\setminus X$.
    Then $X_u\cup X_v$ and $Y_u\cup Y_v$ are independent in~$N$
    because $\{u,v\}$ is independent in $M$.
    Thus we have
    \begin{align*}
        \lambda_{N}(X)
        &\geq \lambda_{N\setminus (E'\setminus (E_u\cup E_v))}(X_u\cup X_v)\\
        &=\rk_{N}(X_u\cup X_v)+\rk_{N}(Y_u\cup Y_v)-\rk_{N}(E_u\cup E_v)\\
        &=\abs{X_u\cup X_v}+\abs{Y_u\cup Y_v}-4\\
        &= \abs{E_u}+\abs{E_v}-4\\
        &=2.
    \end{align*}
    Therefore, by \zcref{lem:robust_partition}, we conclude that $\lambda_N(E_v)=0$ for all $v\in E$.
    However, this implies that $\lambda_M(\{v\})=0$ for all $v\in E$, and consequently $\bw(M)=0$, yielding a contradiction.
\end{proof}

The following lemma shows that 
if two matrices both represent simple matroids,
then their tensor product also represents a simple matroid.

\begin{lemma}[label=lem:tensor_simple]
    Let $M$ and~$N$ be matroids represented by matrices $A$ and $B$, respectively.
    If both $M$ and~$N$ are simple, then $M(A\otimes B)$ is also a simple matroid.
\end{lemma}
\begin{proof}
    Note that $M(A)$ is simple if and only if $A$ has no zero column and no two columns of $A$ are linearly dependent.
    Suppose that $M(A)$ and $M(B)$ are simple. 
    It is easy to observe that $A\otimes B$ has no zero column 
    and no two columns of $A\otimes B$ are linearly dependent.
\end{proof}

\subsection{Constructing a new matroid to multiply \texorpdfstring{$k$}{k} to branch-width}
Our next operation is to take the tensor product with a $k\times (3k-2)$-Vandermonde matrix.
Applying this operation to a matrix representation of a matroid, we can construct a new represented matroid whose branch-width is $k$ times the branch-width of the original matroid.
We start by stating a well-known fact on Vandermonde matrices.

\begin{lemma}[label=lem:Vandermonde_matrix]
    Let $n$ be an integer and let $\mathbb{F}$ be a fixed field.
    Let $\alpha_1,\ldots,\alpha_n$ be distinct elements of $\mathbb F$
    and let
    \[
        A=
        \begin{pmatrix}
            1 & 1 & 1 & \cdots & 1 \\
            \alpha_1 & \alpha_2 & \alpha_3 & \cdots & \alpha_n \\
            \alpha_1^2 & \alpha_2^2 & \alpha_3^2 & \cdots & \alpha_n^2 \\
            \vdots & \vdots & \vdots & \ddots & \vdots \\
            \alpha_1^{n-1} & \alpha_2^{n-1} & \alpha_3^{n-1} & \cdots & \alpha_n^{n-1}
        \end{pmatrix}.
    \]
    Then the columns of $A$ are linearly independent.
\end{lemma}

Let $k\geq 1$ be an integer, let $\mathbb{F}$ be a field with at least $3k-2$ elements, and let $\alpha_1,\ldots,\alpha_{3k-2}\in \mathbb{F}$ be distinct elements.
Let 
\[
    \mathscr{V}_k=\begin{pmatrix}
        1 & 1 & 1 & \cdots & 1 \\
            \alpha_1 & \alpha_2 & \alpha_3 & \cdots & \alpha_{3k-2} \\
            \alpha_1^2 & \alpha_2^2 & \alpha_3^2 & \cdots & \alpha_{3k-2}^2 \\
            \vdots & \vdots & \vdots & \ddots & \vdots \\
            \alpha_1^{k-1} & \alpha_2^{k-1} & \alpha_3^{k-1} & \cdots & \alpha_{3k-2}^{k-1}
    \end{pmatrix}.
\]
By \zcref{lem:Vandermonde_matrix}, $\rk(\mathscr V_k)=k$.

Now let us compute the branch-width of $M(\mathscr{V}_k\otimes A)$.

\begin{proposition}[label=thm:Vandermonde_branchwidth]
    Let $k\geq 1$ be an integer and let $\mathbb{F}$ be a field with at least $3k-2$ elements.
    Let $M$ be a simple represented matroid on $E$ represented by a matrix $A$ over $\mathbb{F}$ and let $N=(E',\mathcal{I}')$ be a matroid represented by~$\mathscr{V}_k\otimes A$. Then the following hold.
    \begin{enumerate}[label=\rm(\alph*)]
\item If $k\ge 2$, then $\bw(N)=k\max\{\bw(M),1\}$.
        \item If $k=1$, then $\bw(N)=\bw(M)$.
    \end{enumerate}
\end{proposition}
\begin{proof}
    If $k=1$, then $\mathscr{V}_1\otimes A=A$, so the statement is trivial.
    Hence from now on, we may assume that~$k\geq 2$.
    Let $F$ be the set of column indices of $\mathscr{V}_k$.
    Then $E'=F\times E$.
    For each $v\in E$, let $E_v=F\times\{v\}$.

Note that $M(\mathscr V_k)$ is a uniform matroid and by \zcref{lem:uniformbw}, $\bw(M(\mathscr V_k))$ is $k$.
    By applying \zcref{lem:tensor} with $B=\mathscr V_k$, 
    we deduce that 
    $\bw(N)\leq \max\{k\bw(M),k\}$.
    Thus it suffices to show that $\bw(N)\geq \max\{k\bw(M),k\}$.

Let $v\in E$ and consider the restriction $N_v$
    of $N$ to~$E_v$. 
    Let $A_v$ be the submatrix of $A$ consisting of the single column corresponding to~$v$.
    Then $N_v=M(\mathscr V_k\otimes A_v)$
    and so $N_v$ is isomorphic to $M(\mathscr V_k)$
    because $A_v$ is non-zero. 
    Therefore $\bw(N_v)\geq \bw(M(\mathscr V_k))= k$.
By \zcref{lem:bw_deletion}, we deduce that $\bw(N)\geq k$.

    Therefore we may assume that $\bw(M)\geq 2$.
    Suppose that there is a branch-decomposition $(T,\mathcal{L})$ of~$N$ of width less than $k\bw(M)$.
    First we would like to show that $E_v$ is titanic with respect to~$\lambda_N$ for each $v\in E$.
    Let $X_1,X_2,X_3$ be pairwise disjoint subsets of $E_v$ such that $X_1\cup X_2\cup X_3=E_v$.
    Since~$\abs{E_v}=3k-2$, there must be some $i\in\{1,2,3\}$ such that $\abs{X_i}\geq k$.
    By \zcref{lem:Vandermonde_matrix}, $X_i$ spans~$E_v$, so we have $\rk_N(X_i)=\rk_N(E_v)$.
    Moreover, since $E'\setminus X_i\supseteq E'\setminus E_v$, we have $\rk_N(E'\setminus E_v)\leq \rk_N(E'\setminus X_i)$.
    Therefore 
    \[\lambda_N(X_i)=\rk_N(X_i)+\rk_N(E'\setminus X_i)-\rk_N(E')\geq \rk_N(E_v)+\rk_N(E'\setminus E_v)-\rk_N(E')= \lambda_N(E_v).\]
    Hence $E_v$ is titanic with respect to $\lambda_N$.
    Thus, $\mathcal{P}=\{E_v\colon v\in E\}$ is a titanic partition of $E'$ satisfying $\lambda_N(E_v)\leq k$ for all $v\in E$.

    Now we use \zcref{lem:titanic_partition} to find a branch-de\-com\-po\-si\-tion~$(S,\mathcal{K})$ of a partitioned matroid~$(N,\mathcal{P})$ with width less than $k\bw(M)$.
    Then by letting $\mathcal{K}'(v)=\mathcal{K}(E_v)$, we obtain a branch-de\-com\-po\-si\-tion~$(S,\mathcal{K}')$ of $M$.
    Let us compute the width of $(S,\mathcal K')$.
    Let $f\in E(S)$ and let $(X_f',Y_f')$ be the partition of~$E$ induced by $\mathcal{K}'^{-1}$ from the components of $S\setminus f$.
    Observe that the submatrix of~$\mathscr{V}_k\otimes A$ induced by taking columns corresponding to~$F\times X_f'$ is equal to $\mathscr V_k\otimes A'$ for a submatrix $A'$ of $A$ induced by taking columns corresponding to~$X_f'$.
    Thus by \zcref{lem:tensor_rank}, we have $\lambda_N^\mathcal{P}(\{E_v:v\in X_f'\})=k\lambda_M(X_f')$, so the width of the edge $f$ in $(S,\mathcal{K})$ is $k$ times the width of~$f$ in $(S,\mathcal{K}')$.
    Hence the width of~$(S,\mathcal{K})$ is $k$ times the width of~$(S,\mathcal{K}')$.
    This implies that the width of~$(S,\mathcal{K}')$ is less than $\bw(M)$, which is a contradiction.
    Therefore, we have $\bw(N)\geq k\bw(M)$.
\end{proof}

Note that we cannot distinguish whether $\bw(M)=0$ or $\bw(M)=1$ by using \zcref{thm:Vandermonde_branchwidth}.
This is the reason why we need \zcref{thm:branchwidth_0_1}.

\subsection{Constructing a new matroid adding \texorpdfstring{$1$}{1} to branch-width}
The last operation we define is an operation that adds $1$ to the branch-width.
Let $\mathbb{F}$ be a field with at least $4$ elements and let $\alpha,\beta,\gamma,\delta\in \mathbb F$ be distinct elements.
For a matrix $A$ over a field $\mathbb{F}$, let $d^+(A)$ be a matrix
\[
    d^+(A)\coloneqq \left( 
        \begin{array}{c|c|c|c} 
            A & A & A & A \\\hline
            \alpha \mathbf{1} & \beta \mathbf{1} & \gamma \mathbf{1} & \delta \mathbf{1}
        \end{array}\right)
\]
that has one more row than $A$, 
where $\mathbf{1}$ represents an all-$1$ matrix with one row.
It is easy to see the following lemma.
\begin{lemma}\label{lem:dplussimple}
    Let $A$ be a matrix over a field $\mathbb F$ with at least four elements.
    If $M(A)$ is simple, then $M(d^+(A))$ is simple.
    \hfill\qedsymbol
\end{lemma}

The following lemma is trivial. We will use it for submatrices of $A$ obtained by removing some columns.
\begin{lemma}\label{lem:rankd}
    Let $A$ be a matrix over a field $\mathbb F$ with at least one column and $\alpha$, $\beta$, $\gamma$, $\delta$ be four distinct elements of $\mathbb F$.
    Then, $\rk (d^+(A))=\rk(A)+1$.
    \hfill\qedsymbol
\end{lemma}

\begin{lemma}[label=lem:plus_1]
    Let $A$ be a matrix over a field $\mathbb{F}$ with at least $4$ elements.
    If $M(A)$ has branch-width $a\geq 1$, then $M(d^+(A))$ has branch-width at most $a+1$.
\end{lemma}
\begin{proof}
    Let $M=M(A)$, $E=E(M)$, and $N=M(d^+(A))$.
    For each $v\in E$, let $E_v$ be the $4$-element subset of $E$ corresponding to columns of $d^+(A)$ having the copy of the column of $v$ in $A$.
    Let~$E'=\bigcup_{v\in E}E_v$.
    By \zcref{lem:rankd}, for all proper  nonempty $X\subseteq E$, we have 
    \[\lambda_N\left( \bigcup_{v\in X}E_v \right)=\lambda_M(X)+1.\]
   
    If $\abs{E}\leq 1$, then it is easy to see that $\bw(N)\leq 2$.
    Thus we may assume that $\abs{E}> 1$.
    Let $(S,\mathcal{K})$ be a branch-decomposition of $M$ witnessing $\bw(M)=a$.
    We construct a branch-decomposition~$(T,\mathcal{L})$ of $N$ as follows.
    Let $T$ be a subcubic tree obtained from the disjoint union of $S$ and $\abs{E}$ copies of rooted binary trees with $4$ leaves by identifying each leaf $\ell=\mathcal{K}(v)$ for $v\in E$ with the root of a rooted binary tree with $4$ leaves.
    We define $\mathcal{L}$ so that for each $v\in E$, four elements in $E_v$ are mapped by $\mathcal{L}$ to the $4$ distinct leaves of the rooted binary tree whose root is identified with~$\mathcal{K}(v)$.

    We claim that $(T,\mathcal{L})$ has width at most $a+1$.
    Let $e\in E(T)$ and let $(X_e,Y_e)$ be the partition of~$E'$ induced by $\mathcal{L}$ and the connected components of $T\setminus e$.
    It is enough to show that $\lambda_N(X_e)\leq a+1$.
    If $X_e= \bigcup_{v\in X}E_v$ for some $X\subseteq E$, then $\lambda_N(X_e)=\lambda_M(X)+1$ because $(X,E\setminus X)$ is a partition induced by $S\setminus e$.
    Thus we may assume that there is some $v\in E$ such that $X_e\cap E_v\neq \emptyset$ and $Y_e\cap E_v\neq \emptyset$.
    Then by the construction of $T$, either $X_e\subseteq E_v$ or $Y_e\subseteq E_v$.
    Without loss of generality, assume that~$X_e\subseteq E_v$.
    Then $\lambda_N(X_e)\leq \rk_N(X_e)\leq 2$, so the width of $e$ is at most $2$.
    This proves that the branch-width of $N$ is at most $a+1$.
\end{proof}

\begin{proposition}[label=thm:branchwidth_plus_1]
    Let $\mathbb{F}$ be a field with at least $4$ elements.
    Let $M$ be a simple represented matroid on~$E$ represented by a matrix $A$ over $\mathbb{F}$ and let $N=(E',\mathcal{I}')$ be a matroid represented by $d^+(A)$.
    If~$\bw(M)\geq 1$, then~$\bw(N)=\bw(M)+1$.
\end{proposition}
\begin{proof}
    The proof proceeds analogously to the proof of \zcref{thm:Vandermonde_branchwidth}. 
    For each $v\in E$, let $E_v$ be the $4$-element subset of~$E$ corresponding to columns of $d^+(A)$ having the copy of the column of~$v$ in~$A$.
    
    By \zcref{lem:plus_1}, $\bw(N)\leq \bw(M)+1$.
    Thus it suffices to show that $\bw(N)\geq \bw(M)+1$.
    
    Let $v\in E$ and let $N_v$ be the restriction of $N$ to~$E_v$.
    Then $N_v$ is isomorphic to a matroid represented by 
    \[ 
    \left( 
        \begin{array}{c|c|c|c} 
            A_v & A_v & A_v & A_v\\\hline
            \alpha & \beta & \gamma & \delta
        \end{array}\right), 
    \] 
    where $A_v$ represents the column vector in $A$ corresponding to the column of~$v$.
    Since $M$ is simple, $A_v$ is non-zero.
    By the elementary row operations, we obtain the matrix
    \[ 
    \left( 
        \begin{array}{c|c|c|c} 
            1 & 1 & 1 & 1\\\hline
            \alpha & \beta & \gamma & \delta
        \end{array}\right), 
    \] 
    which is a matrix of the form $\mathscr V_2$. Thus $\bw(N_v)=\bw(M(\mathscr V_2))=2$.
By \zcref{lem:bw_deletion}, we deduce that $\bw(N)\geq 2$.

    Therefore we may assume that $\bw(M)\geq 2$.
    Suppose that there is a branch-decomposition~$(T,\mathcal{L})$ of $N$ of width at most $\bw(M)$.
    First we would like to show that $E_v$ is titanic with respect to $\lambda_N$ for each $v\in E$.
    Let $X_1,X_2,X_3$ be pairwise disjoint subsets of $E_v$ such that $X_1\cup X_2\cup X_3=E_v$.
    Since $\abs{E_v}=4$, there must be some $i\in \{1,2,3\}$ such that $\abs{X_i}\geq 2$.
    Then $X_i$ spans $E_v$, so we have~$\rk_N(X_i)=\rk_N(E_v)$.
    Moreover, since $E'\setminus X_i\supseteq E'\setminus E_v$, we have $\rk_N(E'\setminus E_v)\leq \rk_N(E'\setminus X_i)$.
    Therefore
    \[ \lambda_N(X_i)=\rk_N(X_i)+\rk_N(E'\setminus X_i)-\rk_N(E')\geq \rk_N(E_v)+\rk_N(E'\setminus E_v)-\rk_N(E')= \lambda_N(E_v). \]
    Hence $E_v$ is titanic with respect to $\lambda_N$.
    Thus, $\mathcal{P}=\{E_v\colon v\in E\}$ is a titanic partition of $E'$ satisfying~$\lambda_N(E_v)\leq 2$ for all $v\in E$.

    Now we use \zcref{lem:titanic_partition} to find a branch-decomposition $(S,\mathcal{K})$ of $(N,\mathcal{P})$ with width at most~$\bw(M)$.
    Then by letting $\mathcal{K}'(v)=\mathcal{K}(E_v)$, we obtain a branch-decomposition $(S,\mathcal{K}')$ of~$M$.
    Let~$e\in E(S)$ and let $(X_e',Y_e')$ be the partition of $E$ induced by $\mathcal{K}'^{-1}$ from the components of~$S\setminus e$.
    Note that the submatrix of $d^+(A)$ induced by $\bigcup_{v\in X_e'}E_v$ is in the form of $d^+(A')$ for a submatrix~$A'$ of~$A$ induced by~$X_e'$.
    Thus we have $\lambda_N^\mathcal{P}(\{E_v:v\in X_e'\})=\lambda_M(X_e')+1$, so the width of $e$ in $(S,\mathcal{K})$ is one plus the width of $e$ in $(S,\mathcal{K}')$.
    This implies that the width of $(S,\mathcal{K}')$ is at most $\bw(M)-1$, which is a contradiction.
    Therefore, we have $\bw(N)\geq \bw(M)+1$.
\end{proof}

\subsection{Combining all the constructions}

We first present a simple lemma.
\begin{lemma}[label=lem:singular_or_simple]
    Given an $n\times n$ matrix $A$ over a field $\mathbb F$, in time $O(n^2)$
    we can either correctly conclude that $A$ is singular
    or correctly conclude that the matroid $M(A)$ is simple, 
    assuming that field elements can be compared in constant time.
\end{lemma}
\begin{proof}
    If some column of $A$ is zero, then $A$ is singular and we stop.
    Otherwise, for each column, we find its uppermost non-zero entry and
    multiply the column by the inverse of that entry.
    This takes time $O(n^2)$, and afterwards the uppermost non-zero entry of
    every column is $1$.
    Since scaling columns changes neither the singularity of~$A$ nor the
    matroid~$M(A)$, we may continue with the scaled matrix.
    By using the trie data structure, in time $O(n^2)$ we can find two
    identical columns if they exist.
    If they exist, then $A$ is singular and we stop.
    Otherwise no two columns are identical and, as the uppermost non-zero
    entry of every column is $1$, no two columns are linearly dependent.
    Since $A$ has no zero column, $M(A)$ is simple.
\end{proof}

We now prove a quadratic-time reduction from \textsc{Matrix Singularity} to \textsc{Branch-width $(2-\varepsilon)$-Approximation for $k$}, 
for every $\varepsilon>0$, 
provided the field is sufficiently large. 
Note that in the following statement, each of $k$, $\alpha$, and~$\mathbb F$ is a fixed constant.

\begin{theorem}[label=thm:new_reduction_general, store=thm:new_reduction_general]
    Let $k\geq 1$ be an integer, 
    $\alpha\geq 2$ be a real, 
and $\mathbb F$ be a field.
If $\abs{\mathbb F}\ge 3p-2$ for every prime divisor~$p$ of $k$, 
    then there is no algorithm that for a matroid~$M$ represented by an input $m\times n$ matrix over~$\mathbb{F}$ with $m\le n$, concludes $\bw(M)>k$ or $\bw(M)< 2k$ in time $O(n^\alpha)$, unless deciding whether an~$n\times n$ matrix over $\mathbb F$ is non-singular can be done in time $O(n^\alpha)$.
\end{theorem}
\begin{proof}
    Suppose that there is an algorithm that for a matroid~$M$ represented by an input $m\times n$ matrix over~$\mathbb{F}$ with $m\le n$, concludes $\bw(M)>k$ or $\bw(M)<2k$ in time $O(n^\alpha)$.
    Let $A$ be an input $n\times n$ matrix over~$\mathbb F$.
    We claim that we can decide whether $A$ is non-singular in time $O(n^\alpha)$.
    By \zcref{lem:singular_or_simple}, in time $O(n^2)$ we either correctly
    conclude that $A$ is singular, in which case we are done, or conclude
    that $M(A)$ is simple.
Thus we may assume that $M(A)$ is simple.

    If $k=1$, then let $B=\left(\begin{smallmatrix} 1 & 0 & 1 \\ 0 & 1 & 1 \end{smallmatrix}\right)\otimes A$.
    If $k\geq 2$, let $k=p_1^{e_1}p_2^{e_2}\cdots p_\ell^{e_\ell}$ be the prime factorization of~$k$ and let
    \[B=\mathscr{V}_{p_1}^{\otimes e_1}\otimes\cdots\otimes \mathscr{V}_{p_\ell}^{\otimes e_\ell}\otimes\begin{pmatrix} 1 & 0 & 1 \\ 0 & 1 & 1 \end{pmatrix} \otimes A\]
     be a $(2kn)\times (3n\prod_{i=1}^\ell (3p_i-2)^{e_i})$ matrix.
    Note that~$
    3kn\le 
    3n\prod_{i=1}^\ell (3p_i-2)^{e_i}\leq 
    3n \prod_{i=1}^\ell (p_i^2)^{e_i} 
    \leq 3k^2n $, so constructing $B$ can be done in time~$O(n^2)$.
    By \zcref{lem:tensor_simple, thm:branchwidth_0_1,thm:Vandermonde_branchwidth}, we have the following dichotomy:
    \begin{itemize}
        \item $\bw(M(B))\geq 2k$ if $\bw(M(A))\geq 1$, and
        \item $\bw(M(B))=k$ if $\bw(M(A))=0$.
    \end{itemize}
    Therefore we can distinguish two cases by using the assumed algorithm in time $O((3k^2n)^\alpha)=O(n^\alpha)$.
\end{proof}

It is worth noting that if the size of the field is large enough, we can reduce the size of the output of the reduction algorithm.
For example, if $\mathbb{F}$ has at least $3k-2$ elements, we can simply use $B=\mathscr{V}_k\otimes \left(\begin{smallmatrix} 1 & 0 & 1 \\ 0 & 1 & 1 \end{smallmatrix}\right)\otimes A$ in the proof of \zcref{thm:new_reduction_general} instead.

If the underlying field is not big enough, then we cannot accommodate $\mathscr{V}_p$ for $p\geq 3$.
In such cases, we can weaken our approximation ratio to $1.5$ as follows.

\begin{theorem}[label=thm:small_field_case, store=thm:small_field_case]
Let $k\geq 0$ be an integer, $\alpha\geq 2$ be a real, and $\mathbb F$ be a field with at least $4$ elements.
    Then there is no algorithm that for a matroid~$M$ represented by an input $m\times n$ matrix over $\mathbb{F}$ with $m\le n$, concludes~$\bw(M)>k$ or $\bw(M)\leq 1.5k$ in time $O(n^\alpha)$, unless deciding whether an~$n\times n$ matrix over $\mathbb{F}$ is non-singular can be done in time $O(n^\alpha)$.
\end{theorem}
\begin{proof}
    The proof proceeds analogously to the proof of \zcref{thm:new_reduction_general}.
    Suppose that there is an algorithm that for a matroid~$M$ represented by an input $m\times n$ matrix~$A$ over $\mathbb{F}$ with $m\le n$, concludes~$\bw(M)>k$ or $\bw(M)\leq 1.5k$ in time $O(n^\alpha)$.
We claim that we can decide whether an~$n\times n$ matrix $A$ over $\mathbb F$ is non-singular in time $O(n^\alpha)$.
    By \zcref{lem:singular_or_simple}, in time $O(n^2)$ we either correctly
    conclude that $A$ is singular, in which case we are done, or conclude
    that $M(A)$ is simple.    
Thus we may assume that $M(A)$ is simple.

    If $k=0$, then by \zcref{lem:branchwidth0}, we are done. Therefore we may assume that $k>0$.

    Let $s=\lfloor \log_2 k\rfloor$.
    Let $a_0,a_1,\ldots,a_s\in \{0,1\}$ be integers such that $k=\sum_{i=0}^s 2^{s-i}a_{i}$ and $a_0=1$.
    Let~$A_0=\left(\begin{smallmatrix} 1 & 0 & 1 \\ 0 & 1 & 1 \end{smallmatrix}\right)\otimes A$.
    For each $i\in [s]$, let \[ 
    A_i=\begin{cases}
        \mathscr{V}_2\otimes A_{i-1}&\text{if } a_i=0,\\ 
        d^+(\mathscr{V}_2\otimes A_{i-1})&\text{if }a_i=1.
    \end{cases}
    \]
    Note that if $A_{i-1}$ has more columns than rows, then so does $A_i$. Therefore the number of rows in~$A_s$ is at most the number of columns of~$A_s$.
    Then $M(A_i)$ is a simple matroid by \zcref{lem:tensor_simple,lem:dplussimple}. 
    It is easy to see the following by induction and \zcref{thm:branchwidth_0_1,thm:Vandermonde_branchwidth,thm:branchwidth_plus_1}.
    \begin{itemize}
        \item  $A_i$ is a $(2^{i+1}n+\sum_{j=1}^i 2^{i-j}a_j)\times (3\cdot 4^{i+\sum_{j=1}^i a_j}n)$ matrix.
        \item $\bw(M(A_i))=\sum_{j=0}^i 2^{i-j}a_j$ if $\bw(M(A))=0$.
        \item $\bw(M(A_i))\geq \sum_{j=0}^i 2^{i-j}a_j+2^i$ if $\bw(M(A))>0$.
    \end{itemize}
    Thus, $A_s$ is a $(2^{s+1}n+k-2^s)\times (3\cdot 4^{s+\sum_{j=1}^s a_j}n)$ matrix such that $\bw(M(A_s))$ is either equal to~$k$ or at least~$k+2^s$.
    Note that $2^{s+1}n+k-2^s = 2^s(2n-1)+k \le k(2n-1)+k \le 2kn$ and~$3\cdot 4^{s+\sum_{j=1}^s a_j } n \le 3\cdot 4^{2s}n \le 3 k^4n$, so constructing $A_s$ can be done in time $O(n^2)$.
    
    If $A$ is non-singular, then $\bw(M(A))=0$ by \zcref{lem:branchwidth0} and so $\bw(M(A_s))=k$. 
    On the other hand, if $A$ is singular, then 
    $\bw(M(A))>0$ and $\bw(M(A_s))\geq k+2^s> 1.5k$, as~$k<2^{s+1}$.
    Therefore we can distinguish two cases by using the assumed algorithm in time $O((3k^4 n)^\alpha)=O(n^\alpha)$.
\end{proof}

\bibliographystyle{amsplain}
\providecommand{\bysame}{\leavevmode\hbox to3em{\hrulefill}\thinspace}
\providecommand{\MR}{\relax\ifhmode\unskip\space\fi MR }
\providecommand{\MRhref}[2]{\href{http://www.ams.org/mathscinet-getitem?mr=#1}{#2}
}
\providecommand{\href}[2]{#2}

\end{document}